\documentclass[submission,copyright,creativecommons]{eptcs}
\providecommand{\event}{AFL 2026} 

\usepackage{iftex, amsthm, amsmath, amsfonts, enumerate, circuitikz, adjustbox, todonotes, breakurl, amssymb, tikz-cd, caption}

\DeclareMathOperator{\symb}{symb}
\DeclareMathOperator{\pref}{pref}
\DeclareMathOperator{\suff}{suff}
\DeclareMathOperator{\rank}{rank}
\newtheorem{theorem}{Theorem}
\newtheorem{proposition}{Proposition}

\newtheorem{corollary}{Corollary}
\newtheorem{definition}{Definition}
\newtheorem{example}{Example}

\title{Irreducibility of Endomorphisms of Finitely Generated Free Semigroups}
\author{Paul C. Bell
\institute{School of Computer Science and Mathematics \\ Liverpool John Moores University \\ Liverpool, UK}
\email{p.c.bell@ljmu.ac.uk}
\and
Eva Foster
\institute{School of Computing and Mathematical Sciences \\ Birkbeck, University of London\\ London, UK}
\email{\quad efoste07@student.bbk.ac.uk}
\and
Daniel Reidenbach
\institute{School of Computing and Mathematical Sciences \\ Birkbeck, University of London\\ London, UK}
\email{\quad d.reidenbach@bbk.ac.uk}
}
\def\titlerunning{Irreducibility of Endomorphisms of Finitely Generated Free Semigroups}
\def\authorrunning{P.~C.~Bell, E.~Foster \& D.~Reidenbach}
\begin{document}
\maketitle
\begin{abstract}
   We introduce and investigate the irreducibility of endomorphisms of finitely generated free semigroups, i.e., we investigate when an endomorphism $\varphi: \Sigma^+ \to \Sigma^+$, where $\Sigma$ is any alphabet, can be nontrivially expressed as a composition $\varphi = \psi_2 \circ \psi_1$ of endomorphisms $\psi_1, \psi_2: \Sigma^+ \to \Sigma^+$. We, hence, study a notion of primality in the endomorphism monoid of the free semigroup---a natural and fundamental concept in this algebraic structure. We establish that irreducibility is a nontrivial property for the class of so-called rank-preserving endomorphisms, and we provide a characteristic condition separating the reducible and irreducible endomorphisms. We also characterise when an endomorphism is a factor of another endomorphism, analyse the non-uniqueness of factorisations of a rank-preserving endomorphism into its irreducible components, and investigate the use of incidence matrices to give insights into the (ir-)reducibility of rank-preserving endomorphisms. 
\end{abstract}

\section{Introduction}
This paper investigates the composition $\psi_2\circ \psi_1$ of two endomorphisms $\psi_1, \psi_2:\Sigma^+\rightarrow \Sigma^+$, defined as $(\psi_2\circ \psi_1)(w)=\psi_2(\psi_1(w))$ for all words $w\in\Sigma^+$. Our primary research question is whether a given endomorphism $\varphi$ can be nontrivially \emph{factorised}, i.\,e., whether there exist endomorphisms $\psi_1, \psi_2$ satisfying $\varphi:\Sigma^+\rightarrow\Sigma^+$. As can be easily demonstrated, this question is particularly interesting for what we refer to as the \emph{rank-preserving} endomorphisms, that is, endomorphisms $\varphi:\Sigma^+\rightarrow\Sigma^+$ for which there exists no set $V\subset\Sigma^+$ with $|V|<|\Sigma|$ such that $\varphi(w)\in V^+$ for all $w\in\Sigma^+$. A rank-preserving endomorphism $\varphi$ is \emph{irreducible} if the equality $\varphi=\psi_2\circ \psi_1$, for rank-preserving endomorphisms $\psi_1, \psi_2$, implies that $\psi_1$ or $\psi_2$ is an automorphism, i.\,e., a renaming. Irreducibility provides a natural notion of primality in the endomorphism monoid of the free semigroup, similar to irreducibility in other algebraic structures.\par
Composition and decomposition of morphisms have been studied in various forms. Ehrenfeucht and Rozenberg~\cite{Ehr78} explore the notion of \emph{simplifiable} and \emph{elementary} homomorphisms, where a morphism $\varphi:\Sigma_1^+\rightarrow\Sigma_3^+$ is simplifiable if it can be expressed as $\varphi=\psi_2\circ \psi_1$ where $\psi_1:\Sigma_1^+\rightarrow\Sigma_2^+$, $\psi_2:\Sigma_2^+\rightarrow\Sigma_3^+$ and $|\Sigma_2|<|\Sigma_1|$. Their work proves that it is decidable to determine if a given morphism is elementary, but does not investigate decompositions without an alphabet reduction (in particular, we note that the rank-reducing endomorphisms are exactly the simplifiable morphisms), nor does it investigate the non-uniqueness of factorisations. Similarly, Berstel, Perrin and Reutenauer~\cite{Reu09} investigate the notion of \emph{indecomposable codes}. In contrast to Ehrenfeucht and Rozenberg, their theory admits an equivalent to prime factorisation under composition, though similarly they do not specify that all alphabets considered must coincide.\par
In the present paper, our restriction to rank-preserving endomorphisms means that the domain, range and all intermediate alphabets in a chain of compositions of endomorphisms coincide, allowing us to consider these chains of compositions more naturally. We answer key questions in this setting and provide, among other things, characterisations of both irreducible rank-preserving endomorphisms under this constraint and when an endomorphism is a factor of another endomorphism, and an exploration of when complete factorisations of endomorphisms into irreducible components are unique.\par
The paper is organised as follows. Section \ref{sec:preliminaries} establishes notation and preliminary results. Section \ref{sec:observations} characterises irreducible endomorphisms and when an endomorphism is a factor of another. Section \ref{sec:uniqueness} explores the (non-)uniqueness of factorisations of endomorphisms, and Section \ref{sec:Incidence matrices} uses incidence matrices to explore the (ir-)reducibility of endomorphisms further.\par
\section{Preliminaries}\label{sec:preliminaries}
By $\mathbb{N}$, we denote the set of all positive integers $\mathbb{N}=\{1, 2, \dots\}$. We define $\mathbb{N}_0=\mathbb{N}\cup \{0\}$. The symbol $\subset$ denotes the \emph{proper subset} relation and $\subseteq$ denotes the \emph{subset} relation. An \emph{alphabet} $\Sigma$ is a non-empty set of symbols. A word (over $\Sigma$) is a finite sequence of symbols from $\Sigma$. The \emph{concatenation} of two words $u$ and $v$, denoted $u\cdot v$, or $uv$ for short, is the word formed by writing the symbols of $u$ followed by the symbols of $v$. The size of a set $A$ is denoted by $|A|$ and the length of a word $w$ by $|w|$. We denote the \emph{empty word}, i.\,e., the unique word of length 0, by $\varepsilon$. The notation $\Sigma^+$ refers to the set of all non-empty words over $\Sigma$, or more precisely, the \emph{free semigroup} generated by $\Sigma$, and $\Sigma^*=\Sigma^+\cup\{\varepsilon\}$, or the \emph{free monoid} generated by $\Sigma$. The result of an $n$-fold concatenation of a word $w$ is denoted by $w^n$, and we call $w^n$ \emph{periodic}, with a period $|w|$. We call a word $v\in\Sigma^*$ a \emph{factor} of a word $w\in\Sigma^*$ if, for some $u_1, u_2 \in \Sigma^*$, we have $w=u_1vu_2$. For a word $w\in\Sigma^*$, a \emph{prefix} $u\in\Sigma^*$ is a word satisfying $w=uv$, for some $v\in\Sigma^*$. Similarly, we call $v\in\Sigma^*$ a \emph{suffix} of $w$ if there exists a $u\in\Sigma^*$ such that $w=uv$. The notations $\pref_n(w)$ and $\suff_m(w)$ refer to the prefix of length $n$ and the suffix of length $m$ of a word $w$ respectively. For any words $v, w\in\Sigma^*$, the notation $|w|_v$ stands for the number of (possibly overlapping) occurrences of $v$ in $w$. For an alphabet $\Sigma=\{a_1, a_2, \dots, a_n\}$, the \emph{Parikh vector} of a word $w\in\Sigma^*$, denoted $p(w)$, is the vector $p(w)=(|w|_{a_1}, |w|_{a_2}, \dots, |w|_{a_n})\in\mathbb{N}_0^n$. The notation $\symb(w)$ represents the set of symbols in the word $w$, with $\symb(w)\subseteq\Sigma$. For a set $W=\{w_1, w_2, \dots, w_n\}\subset\Sigma^+$, we define $\symb(W)=\bigcup_{1\leq i \leq n}\symb(w_i)$. Given a word $w\in\Sigma^*$ that satisfies $w=u_1u_2\cdots u_{n-1}u_n$ with all $u_i\in\Sigma$, the \emph{reversal} of $w$, denoted $w^r$, is defined as $w^r=u_nu_{n-1}\cdots u_2u_1$. Given a non-empty set of words $W\subseteq \Sigma^*$, we call $W$ a \emph{code} if, for any two words $x_1x_2\cdots x_n, y_1y_2\cdots y_m\in W^+$, the equality $x_1x_2\cdots x_n=y_1y_2\cdots y_n$ implies that $n=m$ and $x_i=y_i$ for all $i, 1\leq i \leq n$.\par
A \emph{(semigroup) (homo-)morphism} is a mapping compatible with concatenation, i.\,e., for alphabets $\Sigma_1, \Sigma_2$, we define $\varphi:\Sigma_1^+\rightarrow\Sigma_2^+$ to be a morphism if it satisfies $\varphi(u \cdot v)=\varphi(u)\cdot\varphi(v)$ for all words $u,v\in\Sigma_1^+$---hence a morphism is fully defined as soon as it is declared for all symbols in $\Sigma_1$. A morphism $\varphi:\Sigma_1^+\rightarrow\Sigma_2^+$ is an \emph{endomorphism} if $\Sigma_1=\Sigma_2$, and we write $\varphi:\Sigma^+\rightarrow\Sigma^+$.
The \emph{identity} morphism $\text{id}:\Sigma^+\rightarrow\Sigma^+$ satisfies $\text{id}(a_i)=a_i$ for all $a_i\in\Sigma$. A morphism $\varphi:\Sigma_1^+\rightarrow\Sigma_2^+$ is \emph{injective} if, for all $u, v\in\Sigma_1^+$, the equality $\varphi(u)=\varphi(v)$ implies $u=v$; equivalently, the set $W=\{\varphi(a_i)\mid 1\leq i \leq n\}$ is a code. A morphism $\varphi:\Sigma_1^+\rightarrow\Sigma_2^+$ is a \emph{renaming} if $\varphi$ is injective, and $|\varphi(a_i)|=1$ for every $a_i\in\Sigma_1$. An endomorphism $\varphi:\Sigma^+\rightarrow\Sigma^+$ is an \emph{automorphism} if it is a renaming. For the \emph{composition} of two endomorphisms $\psi_1, \psi_2:\Sigma^+\rightarrow\Sigma^+$, we write $\psi_2\circ \psi_1$, i.\,e., for every $w\in\Sigma^+$, we have $(\psi_2\circ \psi_1)(w)=\psi_2(\psi_1(w))$. If $\varphi=\psi_2\circ \psi_1$, for endomorphisms $\varphi, \psi_1, \psi_2:\Sigma^+\rightarrow \Sigma^+$, we call $\psi_2\circ \psi_1$ a \emph{factorisation} or \emph{decomposition} of $\varphi$. For endomorphisms $\varphi, \mu:\Sigma^+\rightarrow\Sigma^+$, we say that $\mu$ is a \emph{factor} of $\varphi$ if there exist endomorphisms $\psi_1, \psi_2:\Sigma^+\rightarrow\Sigma^+$ satisfying $\varphi=\psi_2 \circ \mu \circ \psi_1$. We say that $\mu$ is a \emph{left-} or \emph{right-factor} of $\varphi$ if there exist endomorphisms $\psi_1, \psi_2:\Sigma^+\rightarrow\Sigma^+$ satisfying $\varphi=\mu\circ \psi_1$ or $\varphi=\psi_2 \circ \mu$ respectively.  For an endomorphism $\varphi:\Sigma^+\rightarrow\Sigma^+$, we define the \emph{reversal} of $\varphi$, denoted $\varphi^r$, as $\varphi^r(a_i)=(\varphi(a_i))^r$ for all $a_i\in\Sigma$.
\section{Observations and Fundamental Characterisations}\label{sec:observations}
In this section, we aim to define formally and characterise reducibility for a rank-preserving endomorphism.
Unlike the factorisation of integers into prime factors, the decomposition of a rank-preserving endomorphism into irreducible factors is not unique. We establish this and other fundamental properties, and provide illustrative examples. \par
It is immediate that all morphisms $\varphi$ can be expressed as  $\varphi=\psi_2\circ \psi_1$ if either $\psi_1$ or $\psi_2$ are automorphisms, and we call such decompositions \emph{trivial}. We begin by showing that morphisms that are not endomorphisms yield no interesting decomposition phenomena.
\begin{proposition}\label{prop: alphabet increase reducible}
    Let $\varphi:\Sigma_1^+\rightarrow\Sigma_2^+$ be a morphism for alphabets $\Sigma_1, \Sigma_2$ satisfying $\Sigma_1\subset\Sigma_2$, and $\Sigma_2=\symb(\{\varphi(a_i) \mid a_i\in\Sigma_1\})$. If $\varphi$ does not satisfy $|\varphi(a_j)|=2$ for exactly one $a_j\in\Sigma_1$ and $|\varphi(a_i)|=1$ for all $a_i\in\Sigma_1\setminus\{a_j\}$, then, there exist an alphabet $\Sigma_3$ satisfying $\Sigma_1\subset\Sigma_3\subseteq\Sigma_2$, a morphism $\psi_1:\Sigma_1^+\rightarrow\Sigma_3^+$ and a non-renaming $\psi_2:\Sigma_3^+\rightarrow\Sigma_2^+$ such that $\varphi=\psi_2\circ\psi_1.$
\end{proposition}
\begin{proof}
    Let $\Sigma_1=\{a_1, a_2, \dots, a_n\}$ and $\Sigma_3=\Sigma_1\cup \{a_{n+1}\}$ be alphabets. Since $\Sigma_1\subset \Sigma_3$, we have $|\Sigma_3|\geq n+1$ and let $a_{n+1}\in\Sigma_3\setminus\Sigma_1$. Since $\Sigma_3\subseteq \Sigma_2$, then $a_{n+1}\in\Sigma_2$. Suppose $|\varphi(a_i)|=1$ for all $a_i\in\Sigma_1$. Then $|\symb(\{\varphi(a_i)\mid a_i\in \Sigma_1\})|\leq n$, contradicting that every symbol of $\Sigma_2$ is used. Therefore, there exists a symbol $a_j\in\Sigma_1$ satisfying $|\varphi(a_j)|\geq 2$. Let $\psi_1:\Sigma_1^+\rightarrow\Sigma_3^+$ be the morphism satisfying $\psi_1(a_j)=a_ja_{n+1}$ and $\psi_1(a_i)=a_i$ for all $a_i\in\Sigma_1\setminus\{a_j\}$. Note that $\psi_1$ is not a renaming. Let $\psi_2:\Sigma_3^+\rightarrow\Sigma_2^+$ be the morphism satisfying $\psi_2(a_i)=\varphi(a_i)$ for all $i\neq j$ and $\psi_2(a_j)=\pref_{m}(\varphi(a_j))$ for $m=\Big\lfloor \frac{|\varphi(a_j)|}{2}\Big\rfloor$ and $\psi_2(a_{n+1})=\suff_{|\varphi(a_j)|-m}(\varphi(a_j))$. As $|\varphi(a_j)|\geq 2$, we note that $\psi_2(a_j)$ and $\psi_2(a_{n+1})$ are well defined. If $|\varphi(a_j)|=2$ and $|\varphi(a_i)|=1$ for all $a_i\in\Sigma_1\setminus\{a_j\}$, then $\psi_2$ is a renaming. Otherwise, the morphism $\psi_2$ is not a renaming and we have found a nontrivial decomposition $\varphi=\psi_2\circ \psi_1$, as required.
\end{proof}
Therefore, we consider only endomorphisms from this point onwards. Our next restriction is to consider \emph{rank-preserving} endomorphisms. A rank-preserving endomorphism cannot be factored through a smaller alphabet than the one it is defined over. 
\begin{definition}\label{def:rank-preserving}
    Let $\varphi:\Sigma^+\rightarrow\Sigma^+$ be an endomorphism for an alphabet $\Sigma$. We call $\varphi$ \emph{rank-preserving} if there is no set $V\subset \Sigma^+$ with $|V|<|\Sigma|$ such that $\varphi(w)\in V^+$ for all $w\in \Sigma^+$. Otherwise, we call $\varphi$ \emph{rank-reducing}.
\end{definition}
By the Defect Theorem \cite{Har04}, non-injective endomorphisms are rank-reducing, since if the set of morphic images of all $a_i\in \Sigma$ is not a code, they can be generated by a set of fewer than $|\Sigma|$ words. The following examples show that injective endomorphisms may still be rank-reducing.
\begin{example}
    Let $\Sigma=\{a, b, c\}$ be an alphabet and $\varphi:\Sigma^+\rightarrow\Sigma^+$ be the endomorphism satisfying $\varphi(a)=aabb, \varphi(b)=baab, \varphi(c)=a$. Note that no word $u\in\{aabb, baab, a\}$ is a suffix of another, so the set $\{aabb, baab, a\}$ is a (suffix-free) code, and $\varphi$ is injective. For all $w\in \Sigma^+$, the set $V=\{a,b\}$ satisfies $\varphi(w)\in V^+$, and so $\varphi$ is rank-reducing. \hfill $\Diamond$
\end{example}
\begin{example}
    Let $\Sigma=\{a, b, c\}$ be an alphabet and $\varphi:\Sigma^+\rightarrow\Sigma^+$ be the endomorphism satisfying $\varphi(a)=ababc, \varphi(b)=abcab, \varphi(c)=ccc$. We may see that $\{ababc, abcab, ccc\}$ is a prefix-free code, therefore $\varphi$ is injective. For all $w\in\Sigma^+$, the set $V=\{ab, c\}$ satisfies $\varphi(w)\in V^+$, and so $\varphi$ is rank-reducing. \hfill $\Diamond$
\end{example}
Many of the examples given in this paper are endomorphisms over a binary alphabet. To verify that such an endomorphism $\varphi:\{a,b\}^+\rightarrow\{a,b\}^+$ is rank-preserving, it is sufficient to  confirm that $\varphi(a)^n\neq \varphi(b)^m$ for all $n,m\in\mathbb{N}$.\par
We can see that rank-reducing endomorphisms are exactly the \emph{simplifiable} morphisms given in \cite{Ehr78}. From this, we can conclude that rank-reducing endomorphisms have a nontrivial factorisation. The notion of simplifiability is shown to be NP-complete in \cite{Ner90}, and therefore the notion of elementariness is shown to be co-NP-complete. Our theory differs from that of simplifiable and elementary morphisms, as we consider the decomposition of \emph{elementary} endomorphisms. We have defined ``rank-preserving'' as an equivalent notion to ``elementary'' to avoid confusion. We now define the notion of reducibility of rank-preserving endomorphisms.
\begin{definition}
    Let $\varphi:\Sigma^+\rightarrow\Sigma^+$ be a rank-preserving endomorphism for an alphabet $\Sigma$. Then $\varphi$ is \emph{reducible} if there exist rank-preserving endomorphisms $\psi_1, \psi_2:\Sigma^+\rightarrow\Sigma^+$ that are not automorphisms such that $\varphi=\psi_2\circ \psi_1$. A rank-preserving endomorphism $\varphi$ is \emph{irreducible} if it is not reducible.
\end{definition}
The following example shows instances of reducible and irreducible endomorphisms.
\begin{example}\label{ex:first irreducible}
    Let $\varphi_1, \varphi_2:\Sigma^+\rightarrow\Sigma^+$ be rank-preserving endomorphisms for $\Sigma=\{a,b\}$, where $\varphi_1(a)=a, \varphi_1(b)=aabb$ and $\varphi_2(a)=a, \varphi_2(b)=babb$. Then, for endomorphisms $\psi_1,\psi_2:\Sigma^+\rightarrow\Sigma^+$ where $\psi_1(a)=a, \psi_1(b)=ab, \psi_2(a)=a$, and $ \psi_2(b)=abb$, it follows that $\varphi_1=\psi_2\circ \psi_1$. It can be easily verified that $\psi_1$ and $\psi_2$ are both rank-preserving and not automorphisms, so $\varphi_1$ is reducible. Using an exhaustive search seen in Appendix \ref{Appendix: sieve}, we can establish that $\varphi_2$ is an irreducible endomorphism. \hfill $\Diamond$
\end{example}
Clearly, verifying irreducibility by exhaustive search is inefficient. We therefore seek a more direct characterisation. Our approach is motivated by the observation that the  structure of an endomorphism $\varphi:\Sigma^+\rightarrow\Sigma^+$ is reflected in $\varphi(w)$, for any word $w\in\Sigma^+$. As an endomorphism is defined by a set of words, each corresponding to the morphic image of an $a_i\in\Sigma$, it is apparent that, for endomorphisms $\psi_1, \psi_2:\Sigma^+\rightarrow\Sigma^+$, satisfying $\varphi=\psi_2\circ \psi_1$, each of the images $\varphi(a_i)$ must reflect the structure of both $\psi_1$ and $\psi_2$. To do this, we introduce the notion of a \emph{factor basis}.
\begin{definition}\label{def:factor basis}
    Let $\Sigma=\{a_1, a_2, \dots, a_n\}$ be an alphabet and let $W\subset \Sigma^+$. Let the set $V\subset\Sigma^+$ satisfy  $|V|\leq |\Sigma|$ and $W\subseteq V^+$. We call such a $V$ a \emph{factor basis of $W$}.\par
    Similarly, given an endomorphism $\varphi:\Sigma^+\rightarrow\Sigma^+$ with $W=\{\varphi(a_i) \mid 1 \leq i \leq n \}$, we call any factor basis $V$ of $W$ a \emph{factor basis of $\varphi$}. We call a factor basis $V$ of $W$ trivial if $V=W$ or $V=\Sigma$. Otherwise, a factor basis is nontrivial.
\end{definition}
It is immediate that every endomorphism $\varphi:\Sigma^+\rightarrow\Sigma^+$ admits the factor bases $V=W$ and $V=\Sigma$, though these may coincide. The factor bases are exactly the sets $V$ referred to in Definition \ref{def:rank-preserving}, and so we may restate this definition in terms of factor bases.
\addtocounter{definition}{-3}
\begin{definition}{(Restated)}
    Let $\varphi:\Sigma^+\rightarrow\Sigma^+$ be an endomorphism for an alphabet $\Sigma$. We call $\varphi$ \emph{rank-preserving} if and only if every factor basis $V$ of $\varphi$ satisfies $|V|=|\Sigma|$.
\end{definition}
\addtocounter{definition}{2}
We can also see that rank-preserving endomorphisms only admit factor bases that are codes, otherwise, by the Defect Theorem \cite{Har04}, there exists a factor basis $V'$ with $|V'|<|V|=|\Sigma|$.\par
As a corollary of the considerations made in \cite{Ner90}, we may restate the decision problem of rank-preserving and rank-reducing in terms of factor bases.
\begin{corollary}{\cite{Ner90}}
    Let $\varphi:\Sigma^+\rightarrow\Sigma^+$ be an endomorphism over an alphabet $\Sigma$.
    \begin{itemize}
        \item Determining if there exists a factor basis of $\varphi$ with cardinality no more than $|\Sigma|-1$ is NP-complete.
        \item Determining whether every factor basis of $\varphi$ has cardinality $|\Sigma|$ is co-NP-complete.
    \end{itemize}\par
\end{corollary}
We now state the first main result of this section, which characterises the reducibility of rank-preserving endomorphisms.
\begin{theorem}\label{thm:d-p reducible}
    Let $\varphi:\Sigma^+\rightarrow\Sigma^+$ be a rank-preserving endomorphism where $\Sigma=$ is an alphabet with $|\Sigma|>1$. Then $\varphi$ is reducible if and only if there exists a nontrivial factor basis $V$ of $\varphi$.
\end{theorem}
\begin{proof}
    Note that any factor basis satisfies $|V|=n$, as $\varphi$ is rank-preserving, and let $\Sigma=\{a_1, a_2, \dots, a_n\}$.\par
    We prove the ``if'' direction first. Assume that such a nontrivial factor basis $V=\{v_1, v_2, \dots, v_n\}$ exists. For each $a_i\in\Sigma$, we have $\varphi(a_i)=v_{i_1}v_{i_2}\cdots v_{i_{j_i}}$ for all $i_{\ell}, 1\leq i_{\ell}\leq n$ and all $j_k\geq 1$. As $V$ is a code, there exists one such factorisation. Then, define $\psi_1(a_i)=a_{i_1}a_{i_2}\cdots a_{i_{j_i}}$, and $\psi_2(a_i)=v_i$ for $1\leq i \leq n$. Thus, $(\psi_2\circ \psi_1)(a_i)=\psi_2(a_{i_1}a_{i_2}\cdots a_{i_{j_i}})=v_{i_1}v_{i_2}\cdots v_{i_{j_i}}=\varphi(a_i)$ as required, and both $\psi_1, \psi_2:\Sigma^+\rightarrow\Sigma^+$ are rank-preserving, and neither $\psi_1$ nor $\psi_2$ is an automorphism. To see this, note that if $\psi_1$ is an automorphism, then $\psi_2(a_i)\in W$ for all $i, 1\leq i \leq n$, and thus $V=W$. Similarly, if $\psi_2$ is an automorphism, then $V=\Sigma$. As $\varphi$ is rank-preserving, there does not exist a factor basis $V'$ of $\varphi$ or $V$ such that $|V'|<|V|$, and so $\psi_2(a_i)=v_i$ is rank-preserving. Similarly, as $\psi_2$ is rank-preserving, there does not exist a factor basis $V''$ of $\psi_2$ such that $|V''|< |\Sigma|$, therefore $\bigcup_{1\leq i \leq n}\symb(\psi_1(a_i))= \Sigma$, and so $\psi_1$ is rank-preserving. Therefore, $\varphi$ is reducible.\par
    We now prove the ``only if'' part of the theorem. Assume that $\varphi$ is reducible, and let $\psi_1, \psi_2:\Sigma^+\rightarrow\Sigma^+$ be rank-preserving endomorphisms satisfying $\varphi=\psi_2\circ \psi_1$, with neither $\psi_1$ nor $\psi_2$ an automorphism. Let $v_i=\psi_2(a_i)$, with $V=\{v_i \mid 1 \leq i \leq n\}$, which by Definition \ref{def:factor basis} is a factor basis. As $\psi_2$ is rank-preserving, we have that $|V|=|\Sigma|$ and $V$ is a code. If $V=\Sigma,$ then $\psi_2$ is an automorphism, contradicting our assumption. If $V=W$, then since $V$ is a code, then we can uniquely decode each $\varphi(a_i)$ with respect to $V$, which must correspond to $\psi_1(a_i)$, since $\varphi=\psi_2\circ \psi_1$. Therefore, and as $\psi_1$ is rank-preserving, each $\psi_1(a_i)$ is a distinct word of length 1 over $\Sigma$, and so $\psi_1$ is an automorphism, contradicting our assumption. Hence, as $V\neq \Sigma$ and $V\neq W$, the factor basis $V$ is nontrivial.
\end{proof}
We now apply Theorem \ref{thm:d-p reducible} to an example morphism.
\begin{example}\label{ex:non unique sec 3}
    Let $\varphi:\{a,b\}^+\rightarrow\{a,b\}^+$ be the rank-preserving endomorphism satisfying $\varphi(a)=ab^3ab^2a$ and $ \varphi(b)=b^2a^2b$. There exists a nontrivial factor basis $V=\{ab, bba\}$ that satisfies the conditions of Theorem \ref{thm:d-p reducible}, and thus $\varphi$ is reducible. We obtain that $\varphi=\psi_2\circ \psi_1$, where $\psi_1(a)=abb, \psi_1(b)=ba, \psi_2(a)=ab$ and $\psi_2(b)=bba$. \hfill $\Diamond$
\end{example}
The following example also illustrates the decision procedure implied by the proof of Theorem \ref{thm:d-p reducible}, and highlights that there may exist multiple nontrivial factor bases for an endomorphism. We can also see that the factor bases given correspond to irreducible rank-preserving endomorphisms.
\begin{example}\label{ex: thm 1 and non-uniquness}
    Let $\varphi:\{a,b\}^+\rightarrow\{a,b\}^+$ be the rank-preserving endomorphism satisfying $\varphi(a)=aba^3ba$ and $ \varphi(b)=a^3(ba)^2$. There exist nontrivial factor bases $V_1=\{a, ab\}$ and $V_2=\{a, ba\}$ that satisfy the conditions of Theorem \ref{thm:d-p reducible}, and so $\varphi$ is reducible. We can also verify that the rank-preserving endomorphisms $\psi_1, \psi_2, \psi_3, \psi_4:\{a,b\}^+\rightarrow\{a,b\}^+$ satisfying $\psi_1(a)=ba^2ba, \psi_1(b)=a^2b^2a, \psi_2(a)=a, \psi_2(b)=ab, \psi_3(a)=aba^2b, \psi_3(b)=a^3b^2, \psi_4(a)=a$ and $\psi_4(b)=ba$, satisfy the equalities $\varphi=\psi_2\circ \psi_1=\psi_4\circ \psi_3$. Note that $\psi_1, \psi_2, \psi_3, \psi_4$ are irreducible, as can be seen using Theorem \ref{thm:d-p reducible}. \hfill $\Diamond$
\end{example}
The next example shows an application of Theorem \ref{thm:d-p reducible} in the case of an irreducible endomorphism.
\begin{example}
    Let $\varphi:\Sigma^+\rightarrow\Sigma^+$ be the rank-preserving endomorphism that, for $\Sigma=\{a,b,c\}$, satisfies $\varphi(a)=aaabcbc, \varphi(b)=acbcc, \varphi(c)=bbabb$. To verify that $\varphi$ is rank preserving, a brute-force search is tedious, so we use a more efficient technique given in Section \ref{sec:Incidence matrices}, and verify that $|P(\varphi)|\neq 0$. We show that $\varphi$ is irreducible by illustrating that there does not exist a nontrivial factor basis $V$ of $\varphi$. Noting that $|\symb(\varphi(c))|=\min\{|\symb(\varphi(a))|, |\symb(\varphi(b))|, |\symb(\varphi(c))|\}$, we choose to begin our considerations with $\varphi(c)$. Informally, this means we consider the image with the `least structure'. We now see that an element $v_1\in V$ must be a non-empty prefix of $\varphi(c)$, so we consider all non-empty prefixes of $\varphi(c)$. If $v_1\in\{bb, bba, bbab, bbabb\}$, then $v_1$ is not a factor of $\varphi(a)$ or $\varphi(b)$, and so it must follow that $\varphi(a), \varphi(b)\in\{v_2, v_3\}^+$. We then consider $v_2\in V$ as a non-empty prefix of $\varphi(b)$. If $v_2\in\{ac, acb, acbc, acbcc\}$, then $v_2$ is not a factor of $\varphi(a)$ or any suffix of $\varphi(c)$.  It then must follow that $\varphi$ admits the trivial factor basis $V_1=\{aaabcbc, acbcc, bbabb\}$. We now consider when $v_1=b$ and $v_2=a$, which means that $v_3=c$. Therefore, the endomorphism $\varphi$ also admits the trivial factor basis $V_2=\{a, b, c\}$. We have shown exhaustively that $\varphi$ only admits trivial factor bases, and by Theorem \ref{thm:d-p reducible}, is irreducible. \hfill $\Diamond$
\end{example}
These examples show us that the factor bases of rank-preserving endomorphisms correspond to a left-factor in its factorisation. In fact, this generalises to all endomorphisms. To prove this, we must first consider the `pre-image' of an endomorphism, given its factor basis. The factor bases of rank-preserving endomorphisms are codes, and thus are uniquely decodable up to a reordering of elements in $V$, so their pre-images consist of exactly the renamings of some word $s$. As the factor bases of rank-reducing endomorphisms can have a smaller cardinality than the alphabet they are defined over, there exists a set of pre-images it may decode to. This motivates the following definition.
\begin{definition}\label{def:tau2}
    Let $\Sigma=\{a_1, a_2, \dots, a_n\}$ be an alphabet, and $V=\{v_1, v_2, \dots, v_m\}\subset \Sigma^+$ be a factor basis of an endomorphism $\varphi:\Sigma^+\rightarrow\Sigma^+$, with $m\leq n$. Choose a surjection $f:\{1, 2, \dots, n\}\rightarrow\{1, 2, \dots, m\}$. For such a surjection, we call the endomorphism $\rho:\Sigma^+\rightarrow\Sigma^+$ given by $\rho(a_i)=v_{f(i)}$ the endomorphism \emph{induced by $V$ and $f$}. For any $w\in\Sigma^+$, we define $\tau_{V}^f(w)$ as the set of words $s\in\Sigma^+$ satisfying $\rho(s)=w$ for $\rho$ induced by $V$ and $f$, noting that this can be empty. We define $\tau_V=\bigcup_{f}\tau_V^f$.
\end{definition}
For a factor basis of a rank-preserving endomorphism, any $f:\{1, 2, \dots, n\}\rightarrow\{1, 2, \dots, n\}$ is a bijection, which allows us to consider a re-ordering of the elements of $V$. The following example illustrates how we can find $\tau_V^f$.
\begin{example}
   Let $\varphi:\{a,b,c\}^+\rightarrow\{a,b,c\}^+$ be the endomorphism satisfying $\varphi(a)=ab, \varphi(b)=ab, \varphi(c)=cbba$. It follows that $\varphi$ admits the factor basis $V=\{ab, cbba\}$. Let $w=abcbbacbbaab$. To consider $\tau_V^f$, we must first choose a surjection $f:\{1, 2, 3\}\rightarrow\{1, 2\}$. We consider three surjections:
    \begin{enumerate}[i)]
        \item If $f_1(1)=1, f_1(2)=1, f_1(3)=2$, then $\tau_V^{f_1}(w)=\{acca, accb, bcca, bccb\}$.
        \item If $f_2(1)=1, f_2(2)=2, f_2(3)=1$, then $\tau_V^{f_2}(w)=\{abba, abbc, cbba, cbbc\}$.
        \item If $f_3(1)=1, f_3(2)=2, f_3(3)=2$, then $\tau_V^{f_3}(w)=\{abba, abca, acba, acca\}$.
    \end{enumerate}
    Note that these surjections are not exhaustive, and the endomorphism $\rho_1:\{a,b,c\}^+\rightarrow\{a,b\}^+$ induced by $f_1$ and $V$ satisfies $\rho_1(a)=ab, \rho_1(b)=ab, \rho_1(c)=cbba$. Furthermore, we can see that $\rho_1=\varphi$. \hfill $\Diamond$
\end{example}
Using this definition, we can show that a factor basis corresponds to a left-factor of any endomorphism.
\begin{proposition}\label{prop:left-factor}
    Let $\varphi:\Sigma^+\rightarrow\Sigma^+$ be an endomorphism for an alphabet $\Sigma=\{a_1, a_2, \dots, a_n\}$, and let $V=\{v_1, v_2, \dots, v_m\}\subset \Sigma^+$ with $m\leq n$. The set $V$ is a factor basis of $\varphi$ if and only if there exist endomorphisms $\psi_1, \psi_2:\Sigma^+\rightarrow\Sigma^+$ with $\varphi=\psi_2\circ \psi_1$ and $V=\{\psi_2(a_i) \mid 1 \leq i \leq n\}$.
\end{proposition}
\begin{proof}
    We begin by proving the ``only if'' direction. Assume that $V=\{v_1, v_2, \dots, v_m\}$ is a factor basis for $\varphi$. By definition, the morphic images satisfy $\varphi(a_i)\in V^+$ for all $i, 1\leq i \leq n$. To define $\psi_2$, we choose a surjection $f:\{1, 2, \dots, n\}\rightarrow \{1, 2, \dots, m\}$ and define $\psi_2(a_i)=v_{f(i)}$, or equivalently, $\psi_2$ is the endomorphism induced by $V$ and some surjection $f$. Since $f$ is surjective, $\psi_2(a_i)\in V^+$ for all $i, 1 \leq i \leq n$, and as $\varphi(a_i)\in V^+$, the set $\tau_V^f$ is non-empty; choose any $s_i\in \tau_V^f(\varphi(a_i))$ and set $\psi_1(a_i)=s_i$. As all $s_i\in\Sigma^+$, we can conclude that $\psi_1$ is an endomorphism. Thus, by our choice of $s_i$, $\psi_2(\psi_1(a_i))=\psi_2(s_i)=\varphi(a_i)$ for each $a_i\in \Sigma$, and $\varphi(a_i)=(\psi_2\circ \psi_1)(a_i)$ as required.\par
    We now prove the ``if'' part of the proposition. We assume that $\varphi=\psi_2\circ \psi_1$ for some endomorphisms $\psi_1, \psi_2:\Sigma^+\rightarrow\Sigma^+$. Let $V=\{\psi_2(a_i)\mid 1\leq i \leq n\}$ and note that $|V|\leq n$. As $\varphi=\psi_2\circ \psi_1$, it follows that $\varphi(a_i)\in V^+$ for all $i, 1\leq i \leq n$, and by Definition \ref{def:factor basis}, $V$ is a factor basis of $\varphi$.
    \end{proof}
Since factorisations of an endomorphism are not necessarily unique (see Example \ref{ex: thm 1 and non-uniquness}), an endomorphism may admit multiple factor bases. To decompose a rank-preserving endomorphism into irreducible components, we must determine which factor bases give the most `refined' factorisation, as shown below.
\begin{definition}
    Given factor bases $V_1=\{v_1, v_2, \dots, v_n\}$ and $V_2=\{v'_1, v'_2, \dots, v'_n\}$, where $V_1\neq V_2$, we say that $V_1$ is \emph{more derived} than $V_2$ if $V_2\subseteq V_1^+$.
\end{definition}
We can represent the hierarchy of derivations of factor bases as a Hasse diagram, where each vertex represents a factor basis, and the directed edges represent the derivation relationship between the factor bases. The graph branches when there exists more than one derivation of a factor basis, or equivalently when the factorisation corresponding to that factor basis is not unique.
\begin{example}\label{ex:factor tree 2}
    Let $\varphi:\{a,b,c\}^+\rightarrow\{a,b,c\}+$ be the rank-preserving endomorphism satisfying $\varphi(a)=abbb, \varphi(b)=acba, \varphi(c)=bbba$. Note that $\varphi$ admits exactly the following factor bases:
    \begin{align*}
V_{1} &= \{abbb, acba, bbba\}, & V_{2} &= \{a, acba, bbb\}, & V_{3} &= \{a, acba, b\}, \\
V_{4} &= \{a, acb, bbb\}, & V_{5} &= \{a, bbb, cba\}, & V_{6} &= \{a, acb, b\}, \\
V_{7} &= \{a, b, cba\}, & V_{8} &= \{a, bbb, cb\}, & V_{9} &= \{a, ac, b\}, \\
V_{10} &= \{a, b, cb\}, & V_{11} &= \{a, b, c\}. \\
\end{align*}
Since $|V_i|=|\{a,b,c\}|$ for all $i, 1\leq i \leq 11$, we have shown that $\varphi$ is rank-preserving.
The derivation relationship between these factor bases is shown in Figure \ref{fig:factor tree 2}. \hfill $\Diamond$
\end{example}
\begin{figure}[h]
    \centering
    \begin{tikzpicture}[>={Stealth[length=1.8mm]},
                        every node/.style={inner sep=1.5pt},
                        x=1cm, y=1cm]
      \node (v1)  at (4.6,3.60) {$V_1$};
      \node (v2)  at (4.6,2.88) {$V_2$};
      \node (v3)  at (0.0,2.16) {$V_3$};
      \node (v4)  at (4.6,2.16) {$V_4$};
      \node (v5)  at (9.2,2.16) {$V_5$};
      \node (v6)  at (0.0,1.44) {$V_6$};
      \node (v7)  at (4.6,1.44) {$V_7$};
      \node (v8)  at (9.2,1.44) {$V_8$};
      \node (v9)  at (2.3,0.72) {$V_9$};
      \node (v10) at (6.9,0.72) {$V_{10}$};
      \node (v11) at (4.6,0.00) {$V_{11}$};
      \draw[->] (v1)--(v2);
      \draw[->] (v2)--(v3);  \draw[->] (v2)--(v4);  \draw[->] (v2)--(v5);
      \draw[->] (v3)--(v6);  \draw[->] (v3)--(v7);
      \draw[->] (v4)--(v6);  \draw[->] (v4)--(v8);
      \draw[->] (v5)--(v7);  \draw[->] (v5)--(v8);
      \draw[->] (v6)--(v9);  \draw[->] (v6)--(v10);
      \draw[->] (v7)--(v10);
      \draw[->] (v8)--(v10);
      \draw[->] (v9)--(v11);
      \draw[->] (v10)--(v11);
    \end{tikzpicture}
    \caption{The Hasse diagram of $\varphi$, given in Example~\ref{ex:factor tree 2}. An arrow points to a more derived factor basis.}
    \label{fig:factor tree 2}
\end{figure}
The next proposition proves the link between a nontrivial factor basis of a rank-preserving endomorphism that is as derived as possible and a factorisation of the endomorphism into irreducible factors. 
\begin{proposition}\label{prop:derivedVset}
 Let $\varphi:\Sigma^+\rightarrow\Sigma^+$ be an endomorphism for an alphabet $\Sigma=\{a_1, a_2, \dots, a_n\}$,  and let $V=\{v_1, v_2, \dots, v_m\}\subset\Sigma^+$ with $m\leq n$ be a factor basis of $\varphi$. The following statements are equivalent:
        \begin{itemize}
            \item The endomorphism $\psi:\Sigma^+\rightarrow\Sigma^+$ induced by the factor basis $V$ and a surjection $f:\{1, 2, \dots, n\}\rightarrow\{1, 2, \dots m\}$ satisfies the equality $\psi=\iota_2\circ \iota_1$ for non-automorphisms $\iota_1, \iota_2:\Sigma^+\rightarrow\Sigma^+$.
            \item There exists a nontrivial factor basis $V'\subset\Sigma^+$ that is more derived than $V$.
        \end{itemize}
    \end{proposition}
    \begin{proof}
         By Proposition \ref{prop:left-factor}, a factor basis corresponds to a left-factor of $\varphi$. If we assume that $\psi=\iota_2\circ \iota_1$, for endomorphisms $\iota_1, \iota_2:\Sigma^+\rightarrow\Sigma^+$, and there exists a factor basis $V'=\{v'_1, v'_2, \dots, v'_{m'}\}\subseteq\Sigma^+$ of $\psi$ that is more derived than $V$, then by definition, $V\subset V'^+$. We then consider $\tau_{V'}^f(v_i)$. If $|\iota_1(a_i)|=1$ for all $i, 1\leq i \leq n$, then $\iota_1$ is a renaming and $V'=V$. Since $V'\neq V$, there exists a $j, \ 1\leq j \leq m$, such that there exists some $s_j\in\tau_{V'}^f(v_j)$ satisfying $|s_j|\geq 2$. Therefore, the endomorphism $\iota_1:\Sigma^+\rightarrow\Sigma^+$ satisfying $\iota_1(a_i)=s_i$ for $i, 1\leq i \leq n$ is not a renaming. We can then see that $\psi=\iota_2\circ \iota_1$, where $\iota_2:\Sigma^+\rightarrow\Sigma^+$ satisfying $\iota_2(a_i)=v'_f(i)$. As $V'\neq \Sigma$, then $\iota_2$ is not an automorphism. Therefore $\psi$ can be nontrivially decomposed as $\psi=\iota_2\circ \iota_1$.\par
        We now assume that $\psi$ can be nontrivially decomposed, and aim to show that there exists a factor basis $V'$ that is more derived than $V$. Let  $\iota_1, \iota_2:\Sigma^+\rightarrow\Sigma^+$ be non-automorphisms such that $\psi=\iota_2\circ \iota_1$. Hence, the endomorphism $\psi$ admits a nontrivial factor basis $V'$ which satisfies $V'=\{v_1', v'_2, \dots, v'_{m'}\}\subset\Sigma^+$ and there exists an endomorphism induced by the factor basis and a surjection $f$ such that $\iota_2(a_i)=v'_{f(i)}$ for $i, \ 1\leq i \leq n$. By the definition of a factor basis, it follows that for $i, \ 1 \leq i \leq n$, the morphic images satisfy $\psi(a_i)\in V'^+$, hence $V\subseteq  V'^+$. We can then conclude that $V'$ is more derived than $V$.
    \end{proof}
    As we aim to find a factorisation of a reducible rank-preserving endomorphism into irreducible components, we now show that irreducibility is preserved under certain basic operations on endomorphisms.
\begin{proposition}\label{prop:reversalreducible}
    Let  $\varphi:\Sigma^+\rightarrow\Sigma^+$ be a rank-preserving endomorphism for an alphabet $\Sigma$. Then $\varphi$ is reducible if and only if $\varphi^r$ is reducible.
\end{proposition}
\begin{proof}
     As $\varphi$ is reducible, it admits a nontrivial factor basis $V=\{v_1, v_2, \dots, v_n\}\subset \Sigma^+$, where $n=|\Sigma|$. We can then see that $\varphi^r$ admits the factor basis $V'=\{v_1^r, v_2^r, \dots, v_n^r\}$, which must also be nontrivial and so $\varphi^r$ is reducible. The converse follows identically, since $(\varphi^r)^r=\varphi$.
\end{proof}
A further operation under which irreducibility is preserved is the equivalence of endomorphisms, which formalises the idea that two endomorphisms are the same, up to composition with an automorphism.
\begin{definition}\label{def: equivalent}
    Let $\varphi, \widetilde{\varphi}:\Sigma^+\rightarrow\Sigma^+$ be endomorphisms for an alphabet $\Sigma$. Then $\varphi$ is \emph{equivalent} to $\widetilde{\varphi}$, denoted $\varphi\cong\widetilde{\varphi}$, if there exist automorphisms $\psi_1, \psi_2:\Sigma^+\rightarrow\Sigma^+$ that satisfy $\varphi=\psi_2\circ \widetilde{\varphi}\circ \psi_1$. 
\end{definition}
\begin{proposition}\label{prop:equivalence relation}
    Equivalence of endomorphisms forms an equivalence relation.
\end{proposition}
\begin{proof}
     We must show that equivalence is reflexive, symmetric and transitive. To prove reflexivity, we note that if the automorphisms $\psi_1, \psi_2:\Sigma^+\rightarrow\Sigma^+$ are equal to the identity morphism, which is an automorphism, then $\varphi=\psi_2\circ \varphi \circ \psi_1$, and $\varphi$ is equivalent to itself. To prove symmetry, since automorphisms are invertible, for automorphisms $\psi_1, \psi_2:\Sigma^+\rightarrow\Sigma^+$, there exist $\psi_1^{-1}, \psi_2^{-1}:\Sigma^+\rightarrow\Sigma^+$ where $\psi_1^{-1}\circ\psi_1=\psi_1 \circ \psi_1^{-1}=\psi_2^{-1}\circ \psi_2 = \psi_2\circ \psi_2^{-1}=\text{id}$. Then, if $\varphi=\psi_2 \circ \widetilde{\varphi}\circ \psi_1$, it follows that $\psi_2^{-1}\circ \varphi \circ \psi_1^{-1}= \psi_2^{-1}\circ \psi_2 \circ \widetilde{\varphi}\circ \psi_1 \circ \psi_1^{-1}$, which gives us $\widetilde{\varphi}=\psi_2^{-1}\circ \varphi \circ \psi_1^{-1}$, and therefore $\widetilde{\varphi}\cong \varphi$. To prove transitivity, we have $\varphi_1\cong \varphi_2$ and $\varphi_2\cong \varphi_3$. This gives us $\varphi_1=\psi_2 \circ \varphi_2 \circ \psi_1$ and $\varphi_2=\psi_2'\circ \varphi_3 \circ \psi_1'$. We can then see that $\varphi_1=\psi_2 \circ \psi_2' \circ \varphi_3 \circ \psi_1' \circ \psi_1$. We note that automorphisms are closed under composition, thus $\varphi_1\cong\varphi_3$.
\end{proof}
\begin{proposition}\label{prop:equivalencereducible}
    Let $\varphi, \widetilde{\varphi}:\Sigma^+\rightarrow\Sigma^+$ be rank-preserving endomorphisms for an alphabet $\Sigma$, where $\varphi\cong \widetilde{\varphi}$. Then $\varphi$ is reducible if and only if $\widetilde{\varphi}$ is reducible.
\end{proposition}
\begin{proof}
    Let $\varphi=\psi_2\circ \psi_1$ be a reducible endomorphism for some non-automorphism, rank-preserving endomorphisms $\psi_1, \psi_2:\Sigma^+\rightarrow\Sigma^+$. Given $\widetilde{\varphi}=\iota_2\circ \varphi\circ \iota_1=\iota_2\circ \psi_2\circ \psi_1\circ \iota_1$, let $\psi_1', \psi_2':\Sigma^+\rightarrow\Sigma^+$ be the endomorphisms defined by $\psi_2'=\iota_2\circ \psi_2$ and $\psi_1'=\psi_1\circ \iota_1$. Both $\psi_1'$ and $\psi_2'$ are rank-preserving and neither is an automorphism, hence $\widetilde{\varphi}$ is reducible. The converse holds by exchanging $\varphi$ and $\widetilde{\varphi}$ in the previous argument.
\end{proof}
We turn to another natural question for the decomposition of an endomorphism $\varphi=\psi_2\circ \psi_1$: given endomorphisms $\varphi, \mu:\Sigma^+\rightarrow\Sigma^+$, is $\mu$ a factor of $\varphi$? In the second main result of this section, we characterise when an endomorphism is a factor of another endomorphism.
\begin{theorem}\label{thm:char_factor_new}
Let $\mu, \varphi:\Sigma^+\rightarrow\Sigma^+$ be endomorphisms for an alphabet $\Sigma$. Then $\mu$ is a factor of $\varphi$ if and only if there exists a factor basis $V = \{v_1, v_2, \ldots, v_m \} \subset \Sigma^+$ of $\varphi$, with $m\leq |\Sigma|$, such that, for some surjection $f:\{1, 2, \dots, |\Sigma|\}\rightarrow\{1, 2, \dots, m\}$ and every $i, 1 \leq i \leq |\Sigma|$, there exists a word $s_i\in\tau_V^f(\varphi(a_i))$ satisfying $s_i\in\{\mu(a_1), \mu(a_2), \dots, \mu(a_{|\Sigma|})\}^+$.
\end{theorem}
\begin{proof} 
    Let $\Sigma=\{a_1, a_2, \dots, a_n\}$. We first prove the ``only if'' direction, i.\,e., we assume that there exist endomorphisms $\psi_1, \psi_2: \Sigma^+ \to \Sigma^+$ satisfying $\varphi = \psi_2 \circ \mu \circ \psi_1$, and let $V=\{\psi_2(a_1),\psi_2(a_2), \dots, \psi_2(a_n)\}$ be a factor basis of $\varphi$. We show that $V$, together with some surjection $f$, has the desired properties. Firstly, we note that the equality $\varphi=\psi_2\circ \mu \circ \psi_1$, and particularly the fact that $\psi_2$ is a left-factor in the factorisation of $\varphi$, directly implies that, for every $i, 1\leq i \leq n$, $\varphi(a_i)$ must be an element of $V^+$, and so $V$ is a factor basis of $\varphi$. Choose a surjection $f:\{1,2, \dots, n\}\rightarrow\{1, 2, \dots, m\}$ such that the endomorphism $\psi_2$ is induced by $V$ and $f$, and therefore is defined by $\psi_2(a_i)=v_{f(i)}$. As a result, there exists a word $s_i\in\tau_V^f(\varphi(a_i))$ such that $\psi_2(s_i)=\varphi(a_i)$. Let $s_i=\mu(\psi_1(a_i))$. Then $s_i\in\{\mu(a_1), \mu(a_2), \dots, \mu(a_n)\}^+$, since $s_i=\mu(\psi_1(a_i))$ and $\psi_2(s_i)=\psi_2(\mu(\psi_1(a_i)))=\varphi(a_i)$, so $s_i\in \tau_V^f(\varphi(a_i))$. Hence, the ``only if'' direction of our theorem holds true.\par
    We now consider the ``if'' direction of the theorem, i.\,e., we assume that there exists a factor basis $V$ of $\varphi$ that satisfies the theorem, and we show that this implies the existence of endomorphisms $\psi_1$ and $\psi_2$ satisfying $\varphi=\psi_2\circ \mu \circ \psi_1$. We define these endomorphisms as follows:
    \begin{itemize}
        \item To determine $\psi_1:\Sigma^+\rightarrow\Sigma^+$, we first consider the factor basis $Z=\{\mu(a_1), \mu(a_2), \dots, \mu(a_n)\}$. Since we postulate that for $i,1\leq i \leq n$ there exists a word $s_i\in\tau_V^f(\varphi(a_i))$ with $s_i\in Z^+$, we may conclude that there exists a surjection $g:\{1, 2, \dots, n\}\rightarrow\{1, 2, \dots, |Z|\}$ there exists a word $s_i'\in\tau_Z^g(s_i)$ satisfying $\mu(s'_i)=s_i$. We then define $\psi_1(a_i)=s_i'$ for $i, 1\leq i \leq n$.
        \item Let $\psi_2:\Sigma^+\rightarrow\Sigma^+$ be the endomorphism induced by the factor basis $V$ and the surjection $f$.
    \end{itemize}
As $V$ is a factor basis of $\varphi$, we know that $\varphi(a_i)=v_{i_1}v_{i_2}\cdots v_{i_{k_i}}$ for some $k_i\geq 1$ and $v_{i_1},v_{i_2},\dots, v_{i_{k_i}}\in V$. Furthermore, our assumptions on $V$ imply that, for all $i$, the endomorphism $\mu$ maps $s_i'$ to a word $s_i=a_{i_1}a_{i_2}\cdots a_{i_{k_i}}$ that is in $\tau_V^f(\varphi(a_i)).$ \par
    We can therefore observe the following: For every $i$, $1 \leq i \leq n$,
    \begin{itemize}
    \item $\psi_1(a_i) = s'_i$,
    \item $\mu(s'_i) = s_i = a_{i_1} a_{i_2} \cdots a_{i_{k_i}}$, 
    \item $\psi_2(a_{i_1} a_{i_2} \cdots a_{i_{k_i}}) = v_{i_1} v_{i_2} \cdots v_{i_{k_i}}$.
    \end{itemize}
    In summary, we have established that $\varphi(a_i) = \psi_2(\mu(\psi_1(a_i)))$ for every $i$ and, hence, that the existence of such a factor basis $V$ of $\varphi$ implies the existence of endomorphisms $\psi_1, \psi_2$ satisfying $\varphi = \psi_2 \circ \mu \circ \psi_1$.     
    \end{proof}
    We now discuss some examples to illustrate Theorem~\ref{thm:char_factor_new} and its proof. These examples focus on applications of Theorem~\ref{thm:char_factor_new} to rank-preserving endomorphisms.
\begin{example}
    Let $\mu, \varphi: \{a,b\}^+ \to \{a,b\}^+$ be the rank-preserving endomorphisms satisfying $\mu(a) = bb$, $\mu(b)= a$, $\varphi(a) = aaaab$, and $\varphi(b) = baaaa$. If we want to use Theorem~\ref{thm:char_factor_new} to establish if $\mu$ is a factor of $\varphi$, we can consider the following potential factor bases $V$ of $\varphi$: $\{a, b\}$, $\{aa, b\}$, $\{aaaa, b\}$ and $\{aaaab, baaaa\}$. We can also observe that $Z = \{\mu(a), \mu(b)\} = \{ bb, a \}$. We choose $V = \{aa, b\}$ and the surjection $f:\{1, 2\}\rightarrow\{1, 2\}$ as $f(1)=2$ and $f(2)=1$, so we can see that the words $bba$ and $abb$ are both in $Z^+$ and satisfy $bba \in \tau_V^f(\varphi(a)) = \tau_V^f(aaaab)$ and $abb \in \tau_V^f(\varphi(b)) = \tau_V^f(baaaa)$. We can, hence, conclude that $\mu$ is a factor of $\varphi$, and the proof of Theorem~\ref{thm:char_factor_new} illustrates that $\varphi = \psi_2 \circ \mu \circ\ \psi_1$, where $\psi_1(a) = ab$, $\psi_1(b) = ba$, $\psi_2(a) = b$, and $\psi_2(b) = aa$. \hfill$\Diamond$
\end{example}
\begin{example}
    Let $\mu, \varphi: \{a,b\}^+ \to \{a,b\}^+$ be the rank-preserving endomorphisms satisfying $\mu(a) = aab$, $\mu(b)= aba$, $\varphi(a) = ababbababbabb$, and $\varphi(b) = abbabbababb$. We now apply Theorem~\ref{thm:char_factor_new} to show that $\mu$ is not a factor of $\varphi$. To this end, we need to consider the set $Z = \{\mu(a),\mu(b)\} = \{aab, aba \}$, and all possible factor bases of $\varphi$.\par
    We first choose $V = \{ab, abb\}$ and then choose our surjection. There exist two surjections $f,g:\{1, 2\}\rightarrow\{1, 2\}$, satisfying $f(1)=1, f(2)=2$ or $g(1)=2, g(2)=1$. Considering each surjection, we obtain $\tau_V^f(\varphi(a)) = \{ababb\}$ and $\tau_V^f(\varphi(b)) = \{bbab\}$, or $\tau_V^g(\varphi(a)) = \{babaa\}$ and $\tau_V^g(\varphi(b)) = \{aaba\}$. We note that neither surjection yields two words that are in $Z^+$, and therefore this choice of $V$ does not satisfy the theorem.\par
    Continuing with $V = \{ ababb, abb\}$, the surjection $f$ (as above) leads to $\tau_V^f(\varphi(a)) = \{aab\}$ and $\tau_V^f(\varphi(b)) = \{bba\}$. In this case, while $aab$ is indeed in $Z^+$, we see that $bba$ is not. The surjection $g$ (as above) implies that $\tau_V^g(\varphi(a)) = \{bba\}$ and $\tau_V^g(\varphi(b)) = \{aab\}$. Hence, we can again note that $bba$ is not in $Z^+$, and therefore our choice of $V$ does not satisfy the theorem.\par
    If $V = \{ab, b\}$, the surjection $f$ implies that $\tau_V^f(\varphi(a)) = \{aabaabab\}$ and $\tau_V^f(\varphi(b)) = \{ababaab\}$. However, neither of these words is in $Z^+$, and the same is true for the surjection $g$, which implies $\tau_V^g(\varphi(a)) = \{bbabbaba\}$ and $\tau_V^g(\varphi(b)) = \{bababba\}$.\par
    The final two candidates are the trivial factor bases of $\varphi$. For $V = \{ \varphi(a), \varphi(b) \}$ it is obvious that $\tau_V^f(\varphi(a))=\{a\}$ and $\tau_V^g(\varphi(a))=\{b\}$, but neither of these words is in $Z^+$. Lastly, if $V = \{a,b\}$, the surjection $f$ $V$ means that $\tau_V^f(\varphi(a)) = \{\varphi(a)\}$ and $\tau_V^f(\varphi(b)) = \{\varphi(b)\}$. However, since, e.\,g., $\varphi(a) = ababbababbabb$ is not in $Z^+$, theorem is not satisfied, and the surjection $g$ does not change this outcome.\par
    These are all possible choices of $V$ of $\varphi$, and none of them satisfy the theorem. We can therefore conclude that $\mu$ is not a factor of $\varphi$. \hfill$\Diamond$
\end{example}
While the decision procedure implied by the proof of Theorem 2 appears complex, the above examples demonstrate that it has a search space that is potentially far smaller than an exhaustive search over all possible
endomorphisms $\psi_1, \psi_2$.
\section{Uniqueness of Factorisations}\label{sec:uniqueness}
In the present section, we further explore the non-uniqueness of the factorisation of a rank-preserving endomorphism into its irreducible components, as discussed in Example \ref{ex: thm 1 and non-uniquness} and Example \ref{ex:factor tree 2}. We identify two forms of non-uniqueness:  non-uniqueness of the irreducible factors, and uniqueness of the irreducible factors but non-uniqueness of factorisations. We shall focus on the non-uniqueness of factors.\par
For a rank-preserving endomorphism $\varphi:\Sigma^+\rightarrow\Sigma^+$ to have non-unique, irreducible factors, the endomorphism $\varphi$ admits distinct maximally derived factor bases (see Proposition \ref{prop:derivedVset}). This occurs when for all $a_i\in\Sigma$, the image $\varphi(a_i)$ lies in the intersection $V_1^+\cap V_2^+$ of the free semigroups generated by two distinct factor bases. For a rank-preserving endomorphism to have unique irreducible factors but a non-unique factorisation into irreducible components, we investigate when endomorphisms commute under composition. \par
A similar question is investigated by Karhum\"{a}ki \cite{Kar84}, who considers intersections of the form $\{x, y\}^*\cap \{u,v\}^*$ for $x, y, u, v\in\Sigma^*$, notably characterising the words that lie in the intersection of monoids of rank 2. We study a related but distinct question: given a rank-preserving endomorphism $\varphi$ inducing the set of words $W=\{\varphi(a_1), \varphi(a_2), \dots, \varphi(a_n)\}$, when do there exist distinct factor bases $V_1, V_2$ of size $n$ such that for all $a_i\in\Sigma$, the morphic images $\varphi(a_i)$ satisfy $\varphi(a_i)\in V_1^+\cap V_2^+$? We provide some conditions for this phenomenon.\par
We begin with some immediate observations on the number of irreducible endomorphisms.
\begin{proposition}\label{prop:unaryirreducible}
         Every rank-preserving endomorphism over a unary alphabet has a unique factorisation into irreducible elements up to reordering.
\end{proposition}
\begin{proof}
     Let $M=\{\varphi \mid \varphi:\{a\}^+\rightarrow\{a\}^+\}$ be the set of rank-preserving endomorphisms over a unary alphabet.
    Every $\varphi\in M$ satisfies $\varphi(a)=a^n$ for some $n\in\mathbb{N}$. The composition of endomorphisms $\psi_1, \psi_2\in M$ that satisfy $\psi_1(a)=a^{n_1}$ and $\psi_2(a)=a^{n_2}$ gives
    $\psi_2\circ \psi_1(a)=a^{n_1n_2}$ for $n_1, n_2\in\mathbb{N}$. It follows that the lengths of irreducible factors of any $\varphi$ correspond to the prime factors of $n$. We can also see that $\psi_1\circ \psi_2=\psi_2\circ \psi_1$, and as multiplication of integers is commutative, the order of the composition of these endomorphisms does not impact the resultant endomorphism. Thus, we can conclude that every endomorphism $\varphi\in M$ has a unique factorisation into irreducible elements up to reordering.
\end{proof}
Thus, we can see that in the unary alphabet, there are infinitely many irreducible endomorphisms.
\begin{corollary}
    The rank-preserving,  irreducible endomorphisms $\varphi:\{a\}^+\rightarrow\{a\}^+$ correspond to those satisfying $\varphi(a)=a^p$, where $p$ is prime.
\end{corollary}
 As we have shown that there are infinitely many rank-preserving irreducible endomorphisms over a unary alphabet, we can make similar observations for larger alphabets by defining a natural embedding of endomorphisms.
    \begin{proposition}\label{prop:morphism embedding}
    Let $\Sigma_1, \Sigma_2$ be alphabets satisfying $\Sigma_1\subset \Sigma_2$. If $\varphi_1:\Sigma_1^+\rightarrow\Sigma_1^+$ is a rank-preserving, reducible endomorphism, then there exists a rank-preserving, reducible endomorphism $\varphi_2:\Sigma_2^+\rightarrow\Sigma_2^+$ such that $\varphi_2(a_i)=\varphi_1(a_i)$ for all $a_i\in\Sigma_1$ and $\varphi_2(a_j)=a_j$ for all $a_j\in\Sigma_2\setminus\Sigma_1$.
\end{proposition}
\begin{proof}
    We begin by defining the embedding of a rank-preserving endomorphism over $\Sigma_1$ to a rank-preserving endomorphism over $\Sigma_2$.
Let $\varphi_1:\Sigma_1^+\rightarrow\Sigma^+_1$ be a rank-preserving endomorphism satisfying $\varphi_1(a_i)=w_i$ for all $a_i\in\Sigma_1$. We extend $\varphi_1$ to a larger alphabet by defining a rank-preserving endomorphism $\varphi_2(a_i)=w_i$ for $a_i\in\Sigma_1$ and $\varphi_2(a_j)=a_j$ for $a_j\in\Sigma_2\setminus\Sigma_1$. If $\varphi_1$ is reducible, then $\varphi_1$ admits a nontrivial factor basis $V_1\subset \Sigma_1^+$, which satisfies $|V_1|=|\Sigma_1|$. It follows that $\varphi_2$ admits the nontrivial factor basis $V_2=V_1\cup (\Sigma_2\setminus \Sigma_1)$, which satisfies $|V_2|=|\Sigma_2|$ and therefore is reducible.
\end{proof}
\begin{corollary}
    For any alphabet $\Sigma$, the set of rank-preserving irreducible endomorphisms over $\Sigma$ is infinite.
\end{corollary}
\begin{proof}
   We may use the embedding defined in Proposition \ref{prop:morphism embedding} for irreducible endomorphisms. Let $\Sigma_1=\{a_1\}$ be an alphabet, and $\varphi_1:\Sigma_1^+\rightarrow \Sigma_1^+$ be an endomorphism $\varphi_1(a_1)=a_1^p$, given for some prime $p$. Let $\Sigma_2$ be an alphabet where $\Sigma_1\subset \Sigma_2$, and define the endomorphism $\varphi_2:\Sigma_2^+\rightarrow\Sigma_2^+$ as $\varphi_2(a_1)=a_1^p$ and $\varphi_2(a_i)=a_i$ for all $a_i\in(\Sigma_2\setminus\{a_1\})$. By Theorem \ref{thm:d-p reducible}, $\varphi_2$ is irreducible. As there exist infinitely many (irreducible) endomorphisms of the form $\varphi_1$, there exist infinitely many irreducible endomorphisms of the form $\varphi_2$ for any alphabet $\Sigma_2$, where $\Sigma_1\subset \Sigma_2$.
\end{proof}
Next, we consider non-uniqueness of the irreducible factors of a binary, rank-preserving endomorphism. We begin by formalising the phenomenon seen in Example \ref{ex: thm 1 and non-uniquness}.
\begin{proposition}\label{prop:aba subword}
        Let  $\varphi:\Sigma^+\rightarrow\Sigma^+$ be a rank-preserving endomorphism and let $\psi_1, \psi_2, \dots, \psi_n:\Sigma^+\rightarrow\Sigma^+$ be irreducible endomorphisms for $\Sigma=\{a,b\}$ that satisfy $\varphi=\psi_n\circ \psi_{n-1} \circ \cdots\circ\psi_1$. The factorisation of $\varphi$ is not unique if there exists an endomorphism $\mu:\Sigma^+\rightarrow\Sigma^+$ satisfying $\mu(a)\in a^+(ba^+)^+$ and $\mu(b)\in a^+(ba^+)^*$, where $\mu$ is a factor of $\varphi$.
    \end{proposition}
    \begin{proof}
         Let $\mu:\Sigma^+\rightarrow\Sigma^+$ be an endomorphism that satisfies $\mu(a)\in a^+(ba^+)^+$ and $\mu(b)\in a^+(ba^+)^*$ and is a factor of $\varphi$. We can see that $\mu$ admits the factor bases $V_1=\{ab, a\}$ and $V_2=\{a, ba\}$. By Proposition \ref{prop:left-factor}, the rank-preserving endomorphisms $\rho_1, \rho_2:\Sigma^+\rightarrow\Sigma^+$ are two left-factors of $\mu$, where $\rho_1(a)=a, \rho_1(b)=ab, \rho_2(a)=a$ and $\rho_2(b)=ba$. Then, for rank-preserving endomorphisms $\iota_1, \iota_2:\Sigma^+\rightarrow\Sigma^+$, it follows that $\mu$ can be expressed as two distinct factorisations $\mu=\rho_1\circ \iota_1=\rho_2\circ \iota_2$. Thus, the factorisation of $\mu$ and therefore the factorisation of $\varphi$ is not unique, as $\mu$ is a factor of $\varphi$.
    \end{proof}
    After we have given sufficient conditions for reducible rank-preserving endomorphisms with non-unique factorisations, we now establish a characteristic condition for a large class of rank-preserving endomorphisms over a binary alphabet to have a unique factorisation into irreducible factors, motivated by the following examples.
\begin{example}
Consider the rank-preserving endomorphism $\varphi:\{a,b\}^+\rightarrow\{a,b\}^+$ that satisfies $\varphi(a)=a^6b^5a^4b^5$ and $\varphi(b)=a^4b^{10}a^4$. We ask if $\varphi$ has a unique factorisation into irreducible factors by considering factor bases of $\varphi$. For $\{v_1, v_2\}$ to be a factor basis of $\varphi$, without loss of generality $v_1$ must be a non-empty prefix of $\varphi(a)$. We first consider when $v_1=a^k$, for $k, 1\leq k \leq 4$. If $k=1$, then all strings of $a$ may be generated by $v_1$, and so we consider how we may generate the strings of $b$. Note that the strings of $b$ have lengths that are multiples of $5$, so $\varphi(a), \varphi(b)\in\{a, b^5\}^+$ and we have found a nontrivial factor basis of $\varphi$. If $k=2$, then we consider that all strings of $a$ have even length, so $\varphi(a), \varphi(b)\in \{a^2, b\}^+$ and we have found a nontrivial factor basis of $\varphi$. We can see by Proposition \ref{prop:derivedVset} that neither $\{a^2, b\}$ nor $\{a, b^5\}$ are more derived than the other---in fact, these factor bases are both maximally derived, and we have shown that $\varphi$ does not have a unique factorisation into irreducible rank-preserving endomorphisms. \hfill $\Diamond$
\end{example}
\begin{example}\label{ex:unique code}
    Consider the rank-preserving endomorphism $\varphi:\Sigma^+\rightarrow\Sigma^+$ for $\Sigma=\{a,b\}$ which satisfies $\varphi(a)=a^2b^5a^3b^5a^2$ and $\varphi(b)=a^5b^5a^3b^5a^7$. We can see that $\varphi$ admits the factor basis $V=\{a, b^5\}$, and investigate if this maximally derived factor basis is unique. To form a factor basis $\{v_1, v_2\}$, we see, without loss of generality, that $v_1$ must be a non-empty prefix of $\varphi(a)$. Exhaustively considering such prefixes of $\varphi(a)$, we see that factor bases of $\varphi$ must satisfy $V_1=\{w_1a^{\alpha_1}, a^{\alpha_2}w_2\}$, $V_2=\{w_3b^{\beta_1}, b^{\beta_2}w_4\}$ or $V_3=\{w_5a^{\gamma_1}, b^{\gamma_2}w_6\}$ for $w_1, w_2, \dots, w_6\in\Sigma^*$, where $\alpha_1, \alpha_2\notin\{2, 3\}, \beta_1+\beta_2=5, \gamma_1\in\{2,3\}$ and $\gamma_2=5$. If $\varphi$ admits any factor basis of the form $V_1$, then $V_1$ is less derived than $V$. We can see this as whenever a $b$ occurs in either element of the factor basis, it must do so as part of a string of length 5 and so $V_1$ does not correspond to another factorisation. If $\varphi$ admits a factor basis of the form $V_2$, then it follows that we cannot generate $\varphi(b)$ from this factor basis, and the same if $\varphi$ admits the factor basis $V_3$, and so the factorisation of $\varphi$ into irreducible components is unique. \hfill $\Diamond$
\end{example}
From these examples, we can see that considering where the boundaries of the elements of the factor bases lie in the morphic images of our rank-preserving endomorphisms can allow us to assert that the factor bases are non-unique. Thus, these examples can be formalised into the following proposition.
\begin{proposition}\label{prop:uniquely reducible}
    Let $\varphi:\Sigma^+\rightarrow\Sigma^+$ be a rank-preserving endomorphism for $\Sigma=\{a,b\}$. Let $\varphi$ satisfy 
        $\varphi(a)= a^{x_1} b^{y_1} a^{x_2} b^{y_2} \cdots a^{x_n} b^{y_n} a^{x_{n+1}},
        \varphi(b) = a^{t_1} b^{z_1} a^{t_2} b^{z_2} \cdots  a^{t_m} b^{z_m} a^{t_{m+1}}$,
 with $x_i, y_i', t_j, z_j'\in\mathbb{N}_{\geq 2}$ for $1\leq i \leq n+1, 1\leq i'\leq n$, $1\leq j \leq m+1, 1\leq j'\leq m$. Let $\gcd(x_1, \dots, x_{n+1}, t_1, \dots, t_{m+1})=k_1$ for $k_1\neq 1$, and $\gcd(y_1, \dots, y_n, z_1, \dots, z_m)=k_2$ for $k_2\neq 1$. It follows that $\varphi$ is uniquely reducible into a composition of irreducible factors if and only if:
    \begin{enumerate}[i)]
        \item $k_1=p_1^{\ell_1}$ for prime $p_1$ and $\ell_1\in\mathbb{N}$, $k_2=1$ and $\varphi(a), \varphi(b)\notin\{a^{\alpha_1}v_1a^{\alpha_2}, a^{\alpha_3}v_2\}^+ $ for any $v_1,v_2\in\Sigma^+$ and $\alpha_1, \alpha_2, \alpha_3\in\mathbb{N}$, or
        \item $k_1=1$, $k_2=p_2^{\ell_2}$ for prime $p_2$, $\ell_2\in\mathbb{N}$ and $\varphi(a), \varphi(b)\notin\{a^{\alpha_1}v_1b^{\beta_1}, b^{\beta_2}v_2a^{\alpha_2}\}^+$ for any $v_1, v_2\in\Sigma^+$ and $\alpha_1, \alpha_2,  \beta_1, \beta_2\in\mathbb{N}$.
    \end{enumerate}
    \end{proposition}
    \begin{proof}
        We begin by considering elements of factor bases of $\varphi$. It is clear that one element of a factor basis must be a prefix of $\varphi(a)$. We consider the three cases where this prefix may end:
    \begin{enumerate}[A)]
        \item the prefix $u\in\Sigma^+$ satisfies $\varphi(a)=ua^{\alpha_1}v_1$, and  $u=a^{\alpha_2}v_2a^{\alpha_3}$ for some $\alpha_1, \alpha_2, \alpha_3 \in\mathbb{N}$, $v_1,v_2\in\Sigma^*$ and $\alpha_3+\alpha_1\in\{x_1, x_2, \dots, x_{n+1}\}$,
        \item the prefix $u\in\Sigma^+$ satisfies $\varphi(a)=ub^{\beta_1}v_1$ and $u=a^{\alpha_1}v_2b^{\beta_2}$ for some $\alpha_1, \beta_1, \beta_2\in\mathbb{N}$, $v_1, v_2\in\Sigma^*$ and $\beta_1+\beta_2\in\{y_1, y_2, \dots, y_n\}$, or
        \item the prefix $u\in\Sigma^+$ satisfies $\varphi(a)=ua^{x_i}v_1$ or $\varphi(a)=ub^{y_i'}v_2$, for $i, 2\leq i \leq n$, $i', 1\leq i' \leq n$, or equivalently the prefix is of the form $a^{\alpha_1}v_3a^{\alpha_2}$ or $a^{\alpha_3}v_4b^{\beta_1}$ respectively, where $\alpha_1, \alpha_2, \alpha_3\in\{x_1, x_2, \dots, x_n\}$ and $\beta_1\in\{y_1, y_2, \dots, y_n\}$ and $v_1, v_2, v_3, v_4\in\Sigma^*$.
    \end{enumerate}
    In case A), we see that the factor basis must be of the form $\{a^{\alpha_1}v_1a^{\alpha_2}, a^{\alpha_3}v_2\}$, for $\alpha_1, \alpha_2, \alpha_3\in\mathbb{N}$, and $v_1, v_2\in\Sigma^*$, where the first element of the factor basis is a prefix of $\varphi(a)$. It is also important to note that every distinct $b^{y_i}$ and $b^{z_j}$ is a factor of $v_1$ or $v_2$. In case B), when the prefix ends in a string over $b$, we can see that the factor basis is of the form $\{a^{\alpha_1}v_1b^{\beta_1}, b^{\beta_2}v_2a^{\alpha_2}\}$ for some $\alpha_1, \alpha_2, \beta_1, \beta_2\in\mathbb{N}$ and $v_1, v_2\in\Sigma^*$. It is also important to note that $\alpha_1=x_1=t_1$ and $\alpha_2=x_{n+1}=t_{m+1}$, otherwise the first and last strings of $a$s would not be covered by the factor basis. Therefore, the length of every string of $a^{x_i}$ or $a^{t_j}$ in the factor basis has a divisor $k_1$. In case C), when the prefix ends between a string over $a$ and a string over $b$, the factor basis is of the form $\{a^{\alpha_1}v_1b^{\beta_1}, a^{\alpha_2}v_2a^{\alpha_3}\}$ or $\{a^{\alpha_1}v_1a^{\alpha_2}, b^{\beta_1}v_2\}$. If these factor bases are not maximally derived, then for $\gcd(x_1, \dots, x_{n+1}, t_1, \dots, t_{m+1})=k_1$ and $\gcd(y_1, \dots, y_n, z_1, \dots, z_m)=k_2$ we have $k_1\neq 1$ or $k_2\neq 1$. Then, the endomorphism $\varphi$ admits factor bases $\{a^{k_1}, b\}$ or $\{a, b^{k_2}\}$ respectively.\par
    We aim to exhaustively consider all values of $k_1$ and $k_2$ to characterise when $\varphi$ is uniquely reducible into a composition of irreducible factors.\par
    We first show that no other factor basis satisfying case C) is maximally derived. As the first element of the factor basis ends at a boundary between a string over $a$ and a string over $b$, every length of the strings over $a$ in the factor basis must be in $\{x_1, x_2, \dots, x_{n+1}, t_1, t_2, \dots, t_{m+1}\}$, and every length of the strings over $b$ in the factor basis must be among $\{y_1, y_2, \dots, y_n, z_1, z_2, \dots, z_n\}$. Since $k_1$ divides all elements of $\{x_1, x_2, \dots, x_{n+1}, t_1, t_2, \dots, t_{m+1}\}$ and $k_2$ divides all elements of $\{y_1, y_2, \dots, y_n, z_1, z_2, \dots, z_n\}$, every factor basis of case C) lies in $\{a^{k_1},b\}^+$ when $k_1\neq 1$ and in $\{a,b^{k_2}\}^+$ when $k_2\neq 1$. Hence $\{a^{k_1},b\}$ or $\{a, b^{k_2}\}$ is more derived than any factor basis satisfying case C). It follows that case C) contributes no new distinct factorisations of $\varphi$.\par
    We next show that no factor basis arising from case B) is maximally derived when $k_1=p_1^{\ell_1}$ and $k_2=1$. We can see that $\varphi$ admits the factor basis $V_1=\{a^{p_1},b\}$, which is maximally derived. We must show that this is the only maximally derived, nontrivial factor basis that $\varphi$ admits. Considering any admissible factor bases of the form B) and C), we note that they must contain all distinct $a^{x_i}$ and $a^{t_j}$. It then must follow that $V_1$ is more derived than any of these factor bases, and so we have not found a new factorisation of $\varphi$. Considering factor bases that satisfy case $A)$, we can see that it is only true that $\varphi$ admits no more factorisations if there does not exist a factor basis that satisfies case $A)$, or equivalently if $\varphi(a),\varphi(b)\notin\{a^{\alpha_1}v_1a^{\alpha_2}, a^{\alpha_3}v_2\}$.\par 
    A similar argument can be made if $k_1=1$ and $k_2=p_2^{\ell_2}$. We can see that $\varphi$ admits the factor basis $V_2=\{a, b^{p_2}\}$, and consider when this is the only maximally derived factor basis that $\varphi$ admits. Considering factor bases satisfying cases A) and C), we note that they contain every distinct $b^{y_i}$ and $b^{z_j}$. It then follows that $V_2$ is more derived than any factor basis of this form. Considering factor bases that satisfy case B) we can see that it is only true that $\varphi$ admits no more factorisations if there does not exist a factor basis that satisfies case B), or equivalently, the morphic images satisfy $\varphi(a), \varphi(b)\notin\{a^{\alpha_1}v_1b^{\beta_1}, b^{\beta_2}v_2a^{\alpha_2}\}$.\par
    If $k_1=p_1^{\ell_1}$ and $k_2=p_2^{\ell_2}$, we can see that $\varphi$ admits the factor bases $V_3=\{a^{p_1}, b\}$ and $V_4=\{a, b^{p_2}\}$. Both $V_3$ and $V_4$ are maximally derived, and so we can see that $\varphi$ does not have a unique factorisation.\par
    If $k_1\neq p_1^{\ell_1}$ and $k_1\neq 1$, then this is equivalent to $k_1$ having more than one distinct prime factor. If we choose two of these prime factors, denoting them $p'_{1}$ and $p''_{1}$, then we see that $\varphi$ admits the factor bases $V_5=\{a^{p'_{1}}, b\}$ and $V_6=\{a^{p''_{1}},b\}$. From this, we can see that $\varphi$ does not have a unique factorisation. The same argument can be made when $k_2\neq p_2^{\ell_2}$ and $k_2\neq 1$.\par
    As we have considered all possible values of $k_1$ and $k_2$, we have shown that $\varphi$ is uniquely reducible into a composition of irreducible factors if and only if it satisfies conditions i) and ii) given in the statement of the proposition.
    \end{proof}
    So far, we have been investigating non-unique factorisations when the factors of the rank-preserving endomorphism are not unique---that is, when an endomorphism admits distinct, maximally derived factor bases. To conclude this section, we now briefly highlight a distinct source of non-uniqueness: when the irreducible factors of a rank-preserving endomorphism are unique, but their order is not. Specifically, we investigate when two endomorphisms commute under composition. \par
Commutativity of morphisms is not a new question---Honkala \cite{Hon19} characterises the commutativity of morphisms which have upper triangular incidence matrices, a concept investigated further in Section \ref{sec:Incidence matrices}. The following proposition gives a simple condition for the commutativity of endomorphisms.
\begin{proposition}\label{prop: commutative}
    Let $\varphi_1, \varphi_2:\Sigma^+\rightarrow\Sigma^+$ be  endomorphisms for an alphabet $\Sigma=\{a_1, a_2, \dots, a_n\}$. Then $\varphi_1 $ and $ \varphi_2$ commute under composition if there exists a set $S\subseteq \Sigma$ such that:
    \begin{enumerate}[i)]
        \item  for all elements $s\in S,\varphi_1(s)\in S^+$ and $\varphi_2(s)=s$, and
        \item for all elements $r\notin S, \varphi_1(r)=r$ and $\varphi_2(r)\in (\Sigma\setminus S)^+$.
    \end{enumerate}
\end{proposition}
\begin{proof}
    Let $\varphi_1, \varphi_2$ satisfy conditions i) and ii). Considering all elements $s\in S$, we have $\varphi_2(s)=s$, so $\varphi_1(\varphi_2(s))=\varphi_1(s)$. As $\varphi_1(s)\in S^+$ and $\varphi_2(s)=s$, we again see that $\varphi_2(\varphi_1(s))=\varphi_1(s)$, so $\varphi_1$ and $\varphi_2$ are commutative for elements $s\in S$. Similarly, for $r\notin S$, we have $\varphi_1(r)=r$, so $\varphi_2(\varphi_1(r))=\varphi_1(r)$, and $\varphi_1(\varphi_2(r))=\varphi_1(r)$, so $\varphi_1$ and $\varphi_2$ are commutative for elements $r\in \Sigma\setminus S$. As $\Sigma=S\cup (\Sigma\setminus S)$, the endomorphisms $\varphi_1$ and $\varphi_2$ commute under composition for all $a_i\in \Sigma$. Note that if $S=\emptyset$ or $S=\Sigma$, then $\varphi_1$ or $\varphi_2$ are the identity morphism respectively.
\end{proof}
We highlight that, given the technical complexity of finding strong results for a specific class of morphisms, as shown in \cite{Hon19}, we expect that obtaining strong results for the commutativity of all endomorphisms will be challenging. However, we consider this a worthwhile open problem.
\section{Incidence Matrices}\label{sec:Incidence matrices}
We now use incidence matrices as a tool to investigate the reducibility of endomorphisms. Incidence matrices are widely applied to the study of morphisms, as seen in \cite{Cau03}, where they are utilised to solve the decreasing length conjecture. We show that endomorphism composition corresponds to matrix multiplication, and so the reducibility of an endomorphism implies the factorisability of its incidence matrix into a product of nontrivial matrices. However, the converse is not true: given a rank-preserving endomorphism $\varphi$ with a factorisable incidence matrix, it can be shown that $\varphi$ is not necessarily reducible. This raises the question of how the structure of an incidence matrix constrains the reducibility of its corresponding rank-preserving endomorphisms. \par
We begin by identifying the relationship between matrix multiplication and the reducibility of endomorphisms. We then identify elementary transformations on matrices that are robust with respect to (ir-)reducibility, and give sufficient conditions for matrices such that every endomorphism they represent is (ir-)reducible. Finally, we highlight how work investigating the factorisation of non-negative integer matrices in their own right can correspond to the reducibility of endomorphisms.
\begin{definition}
    Let $\varphi:\Sigma^+\rightarrow\Sigma^+$ be an endomorphism for an alphabet $\Sigma=\{a_1, a_2, \dots, a_n\}$. We define $P(\varphi)=(m_{i,j})_{1\leq i,j\leq n}$  where $m_{i,j}=|\varphi(a_i)|_{a_j}$, and we call $P(\varphi)$ the \emph{incidence matrix} of $\varphi$. Furthermore, if the endomorphism is not specified, we call any square matrix $P$ with non-negative integer values an incidence matrix.
\end{definition}
To illustrate this definition, we give the following example:
\begin{example}
    Let $\varphi:\{a,b,c\}^+\rightarrow\{a,b,c\}^+$ be the endomorphism that satisfies $\varphi(a)=abcaab, \varphi(b)=babbbba, \varphi(c)=caabccca$. Then $\varphi$ has the incidence matrix $P(\varphi)=\left(\begin{smallmatrix}
        3 & 2 & 1\\
            2 & 5 & 0\\
            3 & 1 & 4
    \end{smallmatrix}\right)$.
    \hfill $\Diamond$
\end{example}
Through our restrictions to rank-preserving endomorphisms, we can see that the incidence matrices we consider are square matrices with no zero rows or columns. The relation between incidence matrices and elementary and simplifiable morphisms has been investigated in~\cite{Bea23}. In this paper, the authors give a sufficient condition for a morphism to be elementary (in our terminology, rank-preserving).
\begin{proposition}[Proposition 4.3 \cite{Bea23}]
    Let $\Sigma_1$ and $\Sigma_2$ be alphabets and $\varphi:\Sigma_1^*\rightarrow\Sigma_2^*$ be an endomorphism. If the rank of $P(\varphi)$ is equal to $|\Sigma_1|$, then $\varphi$ is rank-preserving.
\end{proposition}
Thus, rank-reducing endomorphisms have incidence matrices with a rank strictly smaller than $|\Sigma|$. Note that the converse of this property does not hold; the matrix $P=\left(\begin{smallmatrix} 1 & 1 \\ 1 & 1
\end{smallmatrix}\right)$ with $\rank(P)=1$ satisfies $P(\varphi)=P$ for the rank-preserving endomorphism $\varphi:\{a,b\}^+\rightarrow\{a,b\}^+$ defined as $\varphi(a)=ab, \varphi(b)=ba$. We can also see that this property does not hold for irreducible endomorphisms: a rank-preserving endomorphism that admits an incidence matrix with rank equal to $|\Sigma|$ is not necessarily irreducible, as shown in Example \ref{ex:constant times matrix}.\par
As previously seen in \cite{Cau03}, incidence matrices can be used to calculate the Parikh vector of the image of a word under a given endomorphism. As a corollary of Proposition 4.2 \cite{Cau03}, we can see that multiplication of incidence matrices corresponds to composition of endomorphisms.
\begin{proposition}[Proposition 4.2 \cite{Cau03}]
    Let $\varphi:\Sigma^+\rightarrow\Sigma^+$ be an endomorphism for an alphabet $\Sigma=\{a_1, a_2, \dots, a_n\}$ and let $w\in\Sigma^+$ be a word. It then follows that:
\begin{align*}
       (|\varphi(w)|_{a_1}, |\varphi(w)|_{a_2}, \dots, |\varphi(w)|_{a_n})^T=(|w|_{a_1}, |w|_{a_2}, \dots, |w|_{a_n})^TP(\varphi).
\end{align*}
\end{proposition}
\begin{corollary}\label{cor: multiplication composition}
    Let $\psi_1, \psi_2:\Sigma^+\rightarrow\Sigma^+$ be rank-preserving endomorphisms for $\Sigma=\{a_1, a_2, \dots, a_n\}$ and define $\varphi=\psi_2\circ \psi_1$. Then $P(\varphi)=P(\psi_1)P(\psi_2)$.
\end{corollary}
\begin{proof}
    We can see that $P(\varphi)_{i,j}=|\varphi(a_i)|_{a_j}=\sum_{1\leq k \leq n}|\psi_1(a_i)|_{a_k}|\psi_2(a_k)|_{a_j}=\sum_{1\leq k \leq n}P(\psi_1)_{i,k}P(\psi_2)_{k_j}=(P(\psi_1)P(\psi_2))_{i,j}$.
\end{proof}
This proposition shows that if a rank-preserving endomorphism $\varphi_1$ is reducible, then $P(\varphi_1)$ can be expressed as a product of incidence matrices. The following proposition and example show that the converse is not true---given a rank-preserving endomorphism $\varphi_2$ and its incidence matrix $P(\varphi_2)$, being able to factorise $P(\varphi_2)$ into a product of nontrivial incidence matrices does not necessarily mean that $\varphi_2$ is reducible, merely that there exists a reducible endomorphism $\varphi_3$ such that $P(\varphi_2)=P(\varphi_3)$. This allows us to consider sets of endomorphisms represented by a single incidence matrix, and we ask which matrices force (ir-)reducibility of their respective rank-preserving endomorphisms.\par 
We say an incidence matrix $P$ \emph{represents exclusively (ir-)reducible rank-preserving endomorphisms if every rank-preserving endomorphism $\varphi:\Sigma^+\rightarrow\Sigma^+$ satisfying $P(\varphi)=P$ is (ir-)reducible.}
\begin{proposition}\label{prop:constant times matrix}
    Let $P\in\mathbb{N}_0^{n\times n}$ be an incidence matrix. For a $k\in\mathbb{N}$ with $k>1$, let $P'=kP$. Then $P$ does not represent exclusively irreducible rank-preserving endomorphisms.
\end{proposition}
\begin{proof}
     Multiplying a matrix $M$ by a constant $k$ is equivalent to pre- or post-multiplying $M$ by the matrix $D\in\mathbb{N}^{n\times n}$, where $D=kI_{n}$. The incidence matrix $D$ represents only one endomorphism---the endomorphism $\psi:\Sigma^+\rightarrow\Sigma^+$ such that $\psi(a_i)=a_i^k$ for all $a_i\in \Sigma$. It follows that there exists a rank-preserving endomorphism $\varphi:\Sigma^+\rightarrow\Sigma^+$ with incidence matrix $P$ that can be transformed into one represented by $P'$ by some composition with the endomorphism $\psi$. Thus, the resultant rank-preserving endomorphism is always reducible. However, this would require the rank-preserving endomorphism to satisfy $\varphi(a_j)\in\{a_1^k, a_2^k, \dots, a_n^k\}^+$ for all $a_j\in\Sigma,$ which does not describe all rank-preserving endomorphisms with incidence matrix $P'$---these endomorphisms are not necessarily reducible or irreducible. If $k=1$, then $M$ is multiplied by the identity matrix, which represents the identity morphism, and so we can no longer find a factorisation of $\varphi$.
\end{proof}
The following example illustrates Proposition \ref{prop:constant times matrix}.
\begin{example}\label{ex:constant times matrix}
    Let $M\in\mathbb{N}^{2\times 2}$ be the incidence matrix that satisfies $M=\big(\begin{smallmatrix}  8 & 4 \\ 10 & 6\end{smallmatrix}\big)$. Note that $M=2\left(\begin{smallmatrix}
        4 & 2 \\ 5 & 3
    \end{smallmatrix}\right)$, and let $\varphi_1, \varphi_2:\Sigma^+\rightarrow\Sigma^+$ be the rank-preserving endomorphisms satisfying $\varphi_1(a)=a^2b^2a^4b^2a^2,  \varphi_1(b)=a^6b^2a^2b^4a^2, \varphi_2(a)=aba^3b^2a^2ba^2$ and $\varphi_2(b)=a^5b^3a^2ba^3b^2$. Note that $P(\varphi_1)=P(\varphi_2)=M$. It can be shown that $\varphi_1$ admits the factor basis $V=\{bb, aa\}$, and $\varphi_2$ does not admit a nontrivial factor basis. Thus, the endomorphism $\varphi_1$ is reducible and $\varphi_2$ is irreducible. \hfill $\Diamond$
\end{example}
 We next show how the property of representing exclusively (ir-)reducible rank-preserving endomorphisms when the endomorphisms are rank-preserving can be preserved under elementary transformations of matrices.
\begin{proposition}\label{prop:matrix transformations}
    Let $P\in\mathbb{N}_0^{n\times n}$ be an incidence matrix that represents exclusively (ir-)reducible rank-preserving endomorphisms. Then the following operations are robust with respect to (ir-)reducibility:
    \begin{enumerate}[i)]
        \item swapping the $i^{\text{th}}$ and $j^{\text{th}}$ rows of the matrix and
        \item swapping the $i^{\text{th}}$ and $j^{\text{th}}$ columns of the matrix.
    \end{enumerate}
\end{proposition}
\begin{proof}
    These operations correspond to pre-(or respectively post-) multiplying by a permutation matrix $M$. To find $M$, for $1\leq i,j\leq n$, we swap the $i^{\text{th}}$ and $j^{\text{th}}$ columns of the identity matrix $I^{n\times n}$. In terms of endomorphisms, let $\varphi:\Sigma^+\rightarrow\Sigma^+ $ for $\Sigma=\{a_1, a_2, \dots, a_n\}$ be an irreducible endomorphism. This corresponds to composing $\varphi$ with an automorphism $\psi: \Sigma^+\rightarrow\Sigma^+$ satisfying $\psi(a_i)=a_j, \psi(a_j)=a_i$ and $\psi(a_k)=a_k$ for $i, j, k, \ 1 \leq i, j, k \leq n$ which satisfy $i\neq j$ and $k\neq i,j$.
    By Proposition \ref{prop:equivalencereducible}, these endomorphisms are equivalent as composition with an automorphism does not impact the property of reducibility, and so if $P(\varphi)$ represents exclusively (ir-)reducible rank-preserving endomorphisms, the matrices $MP(\varphi)$ and $P(\varphi)M$ will also hold this property.
\end{proof}
We now begin our investigation into matrices that force (ir-)reducibility of their corresponding rank-preserving endomorphisms. We consider incidence matrices that represent only one endomorphism, which are exactly the monomial matrices. A monomial matrix is a square matrix such that every row and column has exactly one non-zero entry.
\begin{proposition}\label{prop: monomial matrix}
    Let the monomial matrix $P\in\mathbb{N}_0^{n\times n}$ be an incidence matrix. If two of its elements are greater than 1, then $P$ represents exclusively reducible rank-preserving endomorphisms.
\end{proposition}
\begin{proof}
    A monomial matrix represents exactly one rank-preserving endomorphism, which we denote $\varphi:\Sigma^+\rightarrow\Sigma^+$ for an arbitrary alphabet $\Sigma=\{a_1, a_2, \dots, a_n\}$. The endomorphism is some renaming of the rank-preserving endomorphism $\psi:\Sigma^+\rightarrow\Sigma^+$ satisfying $\psi(a_i)=a_i^{m_i}$  for $a_i\in\Sigma$ and $m_i\in\mathbb{N}$. If there exist $m_j, m_k\geq 2$, $\psi$ admits the factor basis $V=\{a_1, a_2, \dots, a_{j-1}, a_j^{m_j}, a_{j+1}, \dots, a_n\}$ for some $a_j\in\Sigma$. As $V\neq \Sigma$ and $\psi(a_k)=a_k^{m_k}$, the factor basis $V$ is nontrivial. By Theorem \ref{thm:d-p reducible}, the rank-preserving endomorphism $\psi$ is reducible, therefore $\varphi$ is reducible.
\end{proof}
As shown in \cite{Hon24}, upper triangular incidence matrices have a more manageable structure in the consideration of endomorphisms. Therefore, upper triangular matrices form a more structured subclass where we can give a sufficient condition for reducibility of rank-preserving endomorphisms using only matrix-level reasoning.
We call a matrix $P\in\mathbb{N}_0^{n\times n}$ \emph{positive upper triangular} if it is upper triangular and $P_{ij}\geq 1$ for all $i\leq j$.
\begin{proposition}\label{prop:upper triangular reducible}
    Let $P\in\mathbb{N}^{n\times n}_0$ be an incidence matrix for $n\geq 3$. If $P$ is positive upper triangular then $P$ represents exclusively reducible rank-preserving endomorphisms.
\end{proposition}
\begin{proof}
     Let $\Sigma=\{a_1, a_2, \dots, a_n\}$ be an alphabet and $\varphi:\Sigma^+\rightarrow\Sigma^+$ be a rank-preserving endomorphism represented by a positive upper triangular matrix. Therefore, we can see that $a_1\in\symb(\varphi(a_1))$ and $a_1\notin\bigcup_{a_j\in\Sigma\setminus\{a_1\}}\symb(\varphi(a_j))$. We can define a factor basis $V=\{ \varphi(a_1), a_2, a_3, \dots, a_{n}\}$, which, by definition satisfies $|V|=n$. Since $|\varphi(a_1)|\geq 2$, $V\neq \Sigma$, and since $|\varphi(a_2)|\geq 2$, we have $V\neq \{\varphi(a_i) \mid 1\leq i \leq n\}$, so $V$ is a nontrivial factor basis and by Theorem \ref{thm:d-p reducible}, all rank-preserving endomorphisms $\varphi$ satisfying $P(\varphi)=P$ are reducible.
\end{proof}
The following proposition illustrates the additional considerations needed for the binary case.
\begin{proposition}\label{prop:upper triangular reducible 2}
        Let $P\in\mathbb{N}^{n\times n}_0$ be an incidence matrix for $n=2$. If $P=\left(\begin{smallmatrix}
            \alpha & \beta \\ 0 & \gamma
        \end{smallmatrix}\right)$ is positive upper triangular and satisfies one of:
        \begin{enumerate}[i)]
            \item $\gamma>1$, 
            \item $\beta=1, \alpha\geq 2$, or
            \item $\alpha=\gamma=1$ and $\beta> 1$,
        \end{enumerate} then $P$ represents exclusively reducible rank-preserving endomorphisms.
\end{proposition}
\begin{proof}
    We treat our three cases separately: \par
For case i), we can see that in the proof of Proposition \ref{prop:upper triangular reducible}, an endomorphism $\varphi:\{a,b\}^+\rightarrow\{a,b\}^+$ with a $2\times 2$ positive upper triangular matrix can admit the factor basis $V=\{\varphi(a),b\}$. This is clearly trivial if $\gamma=1$; thus for $P$ to represent exclusively reducible rank-preserving endomorphisms, we specify that $\gamma> 1$.\par
For case ii), if $\beta=1$ and $\alpha\geq 2$, then the rank-preserving endomorphism $\varphi:\{a,b\}^+\rightarrow\{a,b\}^+$ can admit at least one of the nontrivial factor bases $V_1=\{ab, b\}$ and $V_2=\{ba, b\}$. \par
If an incidence matrix satisfies case iii), then we observe that the matrix $P$ corresponds to the rank-preserving endomorphism $\varphi:\{a,b\}^+\rightarrow\{a,b\}^+$ which satisfies $\varphi(a)=b^kab^m, \varphi(b)=b$, where $k,m\in\mathbb{N}_0$ satisfy $k+m=\beta$. If $k=0$, it follows that the nontrivial factor basis $V_3=\{ab, b\}$ satisfies $\varphi$. If $m=0$, it follows that the nontrivial factor basis $V_4=\{ba, b\}$ satisfies $\varphi$. If $k\neq 0$ and $m\neq 0$ then we have a choice of $V_3$ or $V_4$ as nontrivial factor bases. If $\beta=1$ then the factor bases $V_3$ and $V_4$ would be trivial as they represent $\varphi$ itself. 
\end{proof}
We do not provide a full characterisation of the reducibility of upper triangular, rank-preserving endomorphisms over a binary alphabet as it requires consideration of the structure of the endomorphism which is lost in the construction of the incidence matrices. \par
We now consider incidence matrices that are neither monomial nor upper triangular.
\begin{proposition}\label{prop:exclusively reducible morphisms}
    Let the incidence matrix $P\in\mathbb{N}_0^{2\times 2}$ satisfy one of:
    \begin{enumerate}[i)]
        \item $P=\big(\begin{smallmatrix}
        2 & 2 \\ 2 & 2
    \end{smallmatrix}\big)$, or
    \item $P=\big(\begin{smallmatrix}
         1 & \alpha \\ 1 & \beta
     \end{smallmatrix} \big)$, where $\alpha, \beta >1$ are not coprime.
    \end{enumerate}
    Then $P$ represents exclusively reducible rank-preserving endomorphisms.
\end{proposition}
\begin{proof}
    Beginning with case i), we note that a word $w$ with Parikh-vector $p(w)=(2, 2)$ satisfies $w\in\{aabb, abab, baab, abba, baba, bbaa\}$. Thus, any endomorphism $\varphi:\{a,b\}^+\rightarrow\{a,b\}^+$ with incidence matrix $P$ satisfies one of the following three conditions:
    \begin{enumerate}[A)]
        \item $|\varphi(a)|_{aa}=|\varphi(b)|_{aa}=1$ or $|\varphi(a)|_{bb}=|\varphi(b)|_{bb}=1$,
        \item $\varphi(a)$ or $\varphi(b)$ is periodic, or
        \item $\varphi(a)=abba$ and $\varphi(b)=baab$, or equivalent endomorphisms.
    \end{enumerate}
    If $\varphi$ satisfies condition A) and $\varphi(a)\neq \varphi(b)$, it follows that if $\varphi(a)\neq \varphi(b)$, then $\varphi$ is rank-preserving. We can then see that $\varphi$ admits the nontrivial factor bases $V_1=\{aa, b\}$ or $V_2=\{a, bb\}$ respectively, thus by Theorem \ref{thm:d-p reducible}, is reducible.\par
    If $\varphi$ satisfies condition B), then $\varphi(a)=u_1^2$ or $\varphi(b)=u_2^2$ for $u_1, u_2\in\{a,b\}^2$. If $u_1=u_2$ then $\varphi$ is rank-reducing, which we do not consider. Otherwise, rank-preserving $\varphi$ admits the nontrivial factor bases $V=\{u_1, \varphi(b)\}$ or $V_4=\{\varphi(a), u_2\}$, and by Theorem \ref{thm:d-p reducible} is reducible\par
    If $\varphi$ satisfies condition C), then $\varphi$ is rank-preserving, admits the factor basis $V_3=\{ab, ba\}$, and by Theorem $\ref{thm:d-p reducible}$ is reducible. Therefore, every rank-preserving endomorphism satisfied by $P$ in case i) is reducible.\par
    For case ii), the rank-preserving endomorphisms $\varphi:\{a,b\}^+\rightarrow\{a,b\}^+$ represented by $P$ are of the form $\varphi(a)=b^{k_1}ab^{m_1}$ and $\varphi(b)=b^{k_2}ab^{m_2}$
    for $k_1, k_2, m_1, m_2\in\mathbb{N}_0$ such that $k_1+m_1=\alpha$ and $k_2+m_2=\beta$. If $k_1, k_2, m_1, m_2\neq 0$, then $\varphi$ admits the nontrivial factor bases $V_4=\{ab, b\}$ and $V_5=\{ba, b\}$, and so is reducible. If $k_1=0$ or $k_2=0$ and $m_1, m_2\neq 0$, a nontrivial factor basis of $\varphi$ is $V_4=\{ab, b\}$. If $m_1=0$ or $m_2=0$ and $k_1, k_2\neq 0$, a nontrivial factor basis of $\varphi$ is $V_5=\{ba, b\}$. If $k_1=m_1=0$ or $k_2=m_2=0$, then $P(\varphi)\neq P$. If $k_1=m_2=0$ or $k_2=m_1=0$, then as $\beta$ and $\alpha$ are not coprime, there exists an $\ell\in\mathbb{N}$ such that $\ell \mid \alpha$ and $\ell \mid \beta$ and $\varphi$ admits the nontrivial factor basis $V_6=\{a, b^{\ell}\}$. As we have exhaustively considered all possible values for $k_1, k_2, m_1$ and $m_2$ and have shown that every rank-preserving endomorphism $\varphi$ satisfying $P(\varphi)=P$ has a nontrivial factor basis, and so by Theorem \ref{thm:d-p reducible} is reducible.
\end{proof}
In contrast to the above proposition, we now explore incidence matrices that represent exclusively irreducible rank-preserving endomorphisms.
\begin{proposition}\label{prop:1diag}
    Let the incidence matrix $P\in\mathbb{N}_0^{2\times 2}$ satisfy $P=\big(\begin{smallmatrix}
        1 & y \\ z & 1
    \end{smallmatrix}\big)$ for $y,z\in\mathbb{N}_0$ where $y=1$ and $z=0$ (or vice versa) or $y>1$ and $z>1$. Then $P$ represents exclusively irreducible rank-preserving endomorphisms.
\end{proposition}
\begin{proof}
    The endomorphisms $\varphi:\{a,b\}^+\rightarrow\{a,b\}^+$ represented by the incidence matrix $P$ must be of the form $\varphi(a)=b^{k_1}ab^{k_2}, \varphi(b)=a^{m_1}ba^{m_2}$
    for $k_1, k_2, m_1, m_2\in\mathbb{N}_0$ such that $k_1+k_2=y$ and $m_1+m_2=z$. Note that any endomorphism of this form is rank-preserving. We now consider cases for the elements of the factor basis $V=\{v_1, v_2\}$ of $\varphi$, showing that $V=\{a,b\}$ or $V=\{\varphi(a), \varphi(b)\}$, and thus it must hold true that $V$ is a trivial factor basis.\par
    We first consider the case where $v_1=a^{n_1}$, for $n_1>1$. It follows that $v_1$ is not a factor of $\varphi(a)$ as $|v_1|_a>|\varphi(a)|_a$. For the factor basis to be nontrivial, it must follow that $v_2=\varphi(a)$, but $\varphi(b)\notin \{v_1, v_2\}^+$ as $|v_2|_b>|\varphi(b)|_b$, and $|v_1^k|_b=0$. Hence we obtain a contradiction for a nontrivial factor basis where $v_1=a^{n_1}$ for $n_1>1$, or by the same argument on $\varphi(b)$, where $v_2=b^{n_2}$ for $n_2>1$.\par
    Next, we consider the cases where $v_1$ satisfies $|v_1|_a\geq 1$ and $|v_1|_b\geq 1$. If both $|v_1|_a>1$ and $|v_1|_b>1$, then $v_1$ will not be a factor of either $\varphi(a)$ or $\varphi(b)$ and as neither word is periodic, the factor basis cannot be generated by $v_2$. Without loss of generality, we assume that $|v_1|_a=1$. If $|v_1|_b=1$, then $v_1$ is a factor of both $\varphi(a)$ and $\varphi(b)$. However, it would not be possible that there exists a $v_2$ to create a factor basis of $\varphi$ ---the only symbols in $\varphi(a)$ that are not in the factor $v_1$ are $b$'s and the only symbols in $\varphi(b)$ that are not in the factor $v_1$ are $a$'s.\par
    If $|v_1|_a=1$ and $|v_2|_b>1$, then $v_1$ cannot be a factor of $\varphi(b)$. Therefore, to obtain a factor basis of $\varphi$, we have $v_2=\varphi(b)$ and thus $v_1=\varphi(a)$. The equivalent arguments can be made for $v_2$ if $y=0$. As there do not exist any nontrivial factor bases, it follows that $\varphi$ must be an irreducible endomorphism.
\end{proof}
We have seen that factorisation of an incidence matrix is necessary for reducibility of its associated endomorphisms. Thus, we investigate when a matrix with non-negative integer entries can be expressed as the product of two matrices of the same form. This problem has been investigated in its own right, without the connection to word homomorphisms, in a series of papers by Baeth \cite{Bae22,Bae19,Bae11,Bae15} and Heilbrunn \cite{Hei23}. In these papers, the authors provide necessary conditions for non-negative integer matrices to be reducible. We convert one of these results into a proposition for endomorphisms.\par
Let $\mathbf{a}=(a_1, a_2, \dots, a_n)$ and $\mathbf{b}=(b_1, b_2, \dots, b_n)$ be rows (resp. columns) of a matrix $P$. We say that $\mathbf{a}$ \emph{dominates} $\mathbf{b}$ if $a_i\geq b_i$ for all $a_i\in\mathbf{a}$ and $b_i\in\mathbf{b}$.
\begin{proposition} \label{prop: gcd or minmax}
    Let the incidence matrix $P\in\mathbb{N}_0^{2\times 2}$ satisfy $P=\big(\begin{smallmatrix} \alpha & \beta\\
            \gamma & \delta
    \end{smallmatrix} \big)$,
    such that 
    \begin{enumerate}[i)]
        \item there exists a pair $(i,j)\in \{(\alpha, \beta), (\alpha, \gamma), (\delta, \beta), (\delta, \gamma)\}$ where $\gcd(i,j)\neq 1$, or
        \item a row (resp. column) of $P$ dominates the other row (resp. column) of $P$.
    \end{enumerate}
    Then $P$ does not represent exclusively irreducible rank-preserving endomorphisms.
\end{proposition}
\begin{proof}
    Without loss of generality, we assume that $\gcd(\alpha, \beta)=k$, where $k\neq 1$---symmetric arguments can be made for other cases. We can then note that
       $P=
       \big(\begin{smallmatrix}
           k\alpha' & k\beta'\\
           \gamma & \delta
       \end{smallmatrix}\big)=
       \big(\begin{smallmatrix}
           k & 0\\
           0 & 1
       \end{smallmatrix}\big)
       \big(\begin{smallmatrix}
           \alpha' & \beta'\\
           \gamma & \delta
       \end{smallmatrix}\big).$
   Therefore, there exists a nontrivial decomposition of $P$ and there exists a rank-preserving endomorphism represented by $P$ that is reducible.\par
   If $P=(v_1, v_2)$ for $v_1, v_2\in\mathbb{N}^2$ and $v_1$ dominates $v_2$, then it is easily verifies that $P$ can be factorised as
       $P=\big(\begin{smallmatrix}
           \alpha & \beta\\
           \gamma & \delta
       \end{smallmatrix}\big)=\big(\begin{smallmatrix}
           \alpha-\beta & \beta\\
           \gamma-\delta & \delta
       \end{smallmatrix}\big)\big(\begin{smallmatrix}
           1 & 0\\
           1 & 1
       \end{smallmatrix}\big)$.
      Therefore, there exists a nontrivial decomposition of $P$ and there exists a rank-preserving endomorphism represented by $P$ that is reducible.
\end{proof}
The following example illustrates Proposition \ref{prop: gcd or minmax}.
\begin{example}
    Let $P=\big( \begin{smallmatrix}
        2 & 1 \\ 4 & 7\end{smallmatrix}\big)=\left(\begin{smallmatrix}
            1 & 0 \\ 1 & 1
        \end{smallmatrix}\right)\left(\begin{smallmatrix}
            2 & 1 \\ 2 & 6
        \end{smallmatrix}\right)$ be an incidence matrix. Note that $P$ satisfies the second condition of Proposition \ref{prop: gcd or minmax}. Hence, there exists a reducible rank-preserving endomorphism $\varphi:\{a,b\}^+\rightarrow\{a,b\}^+$ such that $P(\varphi)=P$. Indeed, for  $\varphi(a)=a^2b, \varphi(b)=a^2ba^2b^6$ and, for rank-preserving endomorphisms $\psi_1, \psi_2:\{a,b\}^+\rightarrow\{a,b\}^+$ defined as $\psi_1(a)=a, \psi_1(b)=ab, \psi_2(a)=a^2b, \psi_2(b)=aab^6$, we can see that $\varphi=\psi_2\circ \psi_1$.
\end{example}
Hence, Proposition \ref{prop: gcd or minmax} implies that matrices need to be ``non-coprime'' and ``non-dominated'', i.e. not satisfy both conditions i) and ii) in Proposition \ref{prop: gcd or minmax}, if they represent exclusively irreducible rank-preserving endomorphisms.\par
Together, these results therefore establish a link between the structure of an incidence matrix and the \mbox{(ir-)reducibility} of its corresponding rank-preserving endomorphisms, as well as the benefits and limitations of this approach.
\section{Conclusion}
In this paper, after noting that rank-reducing endomorphisms can always be nontrivially factorised, we have investigated the (ir-)reducibility of rank-preserving endomorphisms in the endomorphism monoid of the free semigroup.  In Theorem \ref{thm:d-p reducible}, we have characterised \emph{reducibility} using factor bases (given in Definition \ref{def:factor basis}) as a key tool. In Theorem \ref{thm:char_factor_new} we have continued to characterise when, given endomorphisms $\mu$ and $\varphi$, the endomorphism $\mu$ is a \emph{factor} of $\varphi$. Furthermore, we have shown that the factorisation of a rank-preserving endomorphism into irreducible components is not unique, and how we can identify which factor bases correspond to the most derived factorisation. In Section \ref{sec:uniqueness}, we have given some necessary and sufficient conditions for a large class of rank-preserving binary endomorphisms to have a unique factorisation. In Section \ref{sec:Incidence matrices}, we have used incidence matrices to investigate the factorisation of rank-preserving endomorphisms, extending results already given for non-negative integer matrices. The most immediate open problems are:
\begin{enumerate}[i)]
    \item determining characteristic conditions for rank-preserving endomorphisms to have a unique factorisation into irreducible components,
    \item extending the investigations of uniqueness of factorisations and incidence matrices to larger alphabet sizes,
    \item establishing the computational complexity of the decision procedure arising from the proof of Theorem \ref{thm:char_factor_new}.
\end{enumerate}
\bibliographystyle{eptcs}
\bibliography{generic}

@article{Ehr78,
    author = {A.~Ehrenfeucht and G.~Rozenberg},
    title ={Simplifications of homomorphisms} ,
    journal = {Information and Control},
    volume={38},
    number={3},
    pages={298--309},
    doi={10.1016/S0019-9958(78)90095-5},
    year = {1978}
}

@book{Reu09,
    author = {C.~Reutenauer and D.~Perrin and J.~Berstel} ,
    title = {Codes and Automata},
    publisher = {Cambridge University Press} ,
    collection={Encyclopedia of Mathematics and its Applications},
    doi={10.1017/CBO9781139195768},
    year = {2009}
}

@article{Kar84,
    author = {J.~Karhum{\"{a}}ki},
    title = {A note on intersections of free submonoids of a free monoid} ,
    journal = {Semigroup Forum} ,
    volume = {29},
    issue={1},
    pages= {183--205},
    doi={10.1007/BF02573324},
    year = {1984}
}

@mastersthesis{Hei23,
    author = {G.~Heilbrunn},
    title = {Decomposition of Nonnegative Integer-entry Matrices},
    school = {Franklin and Marshall College},
    year = {2023},
    url={https://digital.fandm.edu/islandora/decomposition-nonnegative-integer-entry-matrices}
}

@article{Cau03,
    author = {S.~Cautis and F.~Mignosi and J.~Shallit and M.~Wang and S.~Yazdani},
    title = {Periodicity, morphisms and matrices},
    journal = {Theoretical Computer Science},
    volume ={295},
    pages={107--121},
    number={1},
    doi={10.1016/S0304-3975(02)00398-5},
    year = {2003}
}

@article{Hon19,
    author = {J.~Honkala},
    title = {A characterization of free pairs of upper triangular free monoid morphisms},
    journal = {Information and Computation},
    volume={267},
    pages={110--115},
    doi={10.1016/j.ic.2019.03.007},
    year = {2019}
}

@article{Hon24,
    author = {J.~Honkala},
    title = {Commuting upper triangular binary morphisms} ,
    journal = {Fundamenta Informaticae} ,
    url        = {https://fi.episciences.org/10954},
    doi        = {10.46298/fi.10954},
    volume={191, Issues 3--4: Iiro Honkala's 60 Birthday},
    issue={3--4},
    pages={285--298},
    year = {2024}
}

@article{Bea23,
    author = {M.-P.~ B\'{e}al and D. Perrin and A. Restivo},
    title = {Recognizability of morphisms} ,
    journal = {Ergodic Theory and Dynamical Systems} ,
    volume={43},
    number={11},
    Doi={10.1017/etds.2022.109},
    pages={3578--3602},
    year = {2023}
}

@article{Bae22,
    author = {N.~R.~Baeth and H.~Chen and G.~Heilbrunn and R.~Liu and M.~Young},
    title = {Semigroups of non-negative integer-valued matrices},
    journal ={Communications in Algebra} ,
    volume={50},
    number={3},
    pages={1199--1219},
    doi={10.1080/00927872.2021.1979569},
    year = {2022}
}

@article{Bae11,
    author = {N.~R. ~Baeth and V.~Ponomarenko and D.~ Adams and R.~Ardila and D.~Hannasch and A.~Kosh and H.~McCarthy and R.~Rosenbaum},
    title = {Number theory of matrix semigroups},
    journal = {Linear Algebra and its Applications},
    volume={434},
    number={3},
    pages={694--711},
    doi={10.1016/j.laa.2010.09.028},
    year = {2011}
}

@article{Bae15,
    author = {N.~R.~Baeth and D.~Smertnig},
    title = {Factorization theory: From commutative
to noncommutative settings},
    journal = {Journal of Algebra},
    volume={441},
    pages={475--551},
    doi={10.1016/j.jalgebra.2015.06.007},
    year = {2015}
}

@article{Bae19,
    author = {N.~R.~Baeth and M.~Enlow},
    title = {Multiplicative factorization in numerical semigroups},
    journal = {International Journal of Algebra and Computation},
    volume={30},
    number={2},
    doi={10.1142/S0218196720500058},
    year ={2020},
    pages={419--430}
}

@article{Har04,
    author = {T.~ Harju and J. Karhum{\"{a}}ki},
    title = {Many aspects of defect theorems},
    journal = {Theoretical Computer Science},
    volume={324},
    number={1},
    pages={35--54},
    doi={10.1016/j.tcs.2004.03.051},
    year = {2004}
}

@article{Ner90,
    author = {J.~N\'{e}raud},
    title = {Elementariness of a finite set of words is co-NP-complete},
    journal = {RAIRO-Theor. Inf. Appl.},
    year = {1990},
    volume = {24},
    number = {5},
    doi={10.1051/ita/1990240504591},
    pages = {459--470}
}
\newpage
\appendix
\section{ Exhaustive Search} \label{Appendix: sieve}
\begin{table}[h]
  \centering
  \tiny
  \begin{adjustbox}{center,max width=\linewidth}
    \begin{tabular}{||c||c|c|c|c|c|c||}
      \hline
      \hline
     & $\psi_1(a)=a$ &  $\psi_1(a)=a$ & $\psi_1(a)=a$ 
        & $\psi_1(a)=a$ & $\psi_1(a)=a$ & $\psi_1(a)=a$ \\
      & $\psi_1(b)=b$ & $\psi_1(b)=ba$ & $\psi_1(b)=ab$ & $\psi_1(b)=bb$
    & $\psi_1(b)=bbb$ & $\psi_1(b)=bba$ \\
      \hline
      \hline
   $\psi_2(a)=a$ & $\varphi(a)=a$ & $\varphi(a)=a$ & $\varphi(a)=a$ & $\varphi(a)=a$ & $\varphi(a)=a$ & $\varphi(a)=a$ \\
  $\psi_2(b)=b$  & $\varphi(b)=b$ & $\varphi(b)=ba$ & $\varphi(b)=ab$ & $\varphi(b)=bb$ & $\varphi(b)=bbb$ & $\varphi(b)=bba$ \\
  \hline
  $\psi_2(a)=a$ & $\varphi(a)=a$ & $\varphi(a)=a$ & $\varphi(a)=a$ & $\varphi(a)=a$ & $\varphi(a)=a$ & $\varphi(a)=a$ \\
  $\psi_2(b)=ba$  & $\varphi(b)=ba$ & $\varphi(b)=baa$ & $\varphi(b)=aba$ & $\varphi(b)=baba$ & $\varphi(b)=bababa$ & $\varphi(b)=babaa$ \\
  \hline
    $\psi_2(a)=a$ & $\varphi(a)=a$ & $\varphi(a)=a$ & $\varphi(a)=a$ & $\varphi(a)=a$ & $\varphi(a)=a$ & $\varphi(a)=a$ \\
  $\psi_2(b)=ab$  & $\varphi(b)=ab$  & $\varphi(b)=aba$ & $\varphi(b)=aab$ & $\varphi(b)=abab$ & $\varphi(b)=ababab$ & $\varphi(b)=ababa$ \\
  \hline
    $\psi_2(a)=a$ & $\varphi(a)=a$ & $\varphi(a)=a$ & $\varphi(a)=a$ & $\varphi(a)=a$ & $\varphi(a)=a$ & $\varphi(a)=a$ \\
  $\psi_2(b)=bb$  & $\varphi(b)=bb$ & $\varphi(b)=bba$ & $\varphi(b)=abb$ & $\varphi(b)=bbbb$ & $\varphi(b)=bbbbbb$ & $\varphi(b)=bbbba$ \\
  \hline
    $\psi_2(a)=a$ & $\varphi(a)=a$ & $\varphi(a)=a$ & $\varphi(a)=a$ & $\varphi(a)=a$ & $\varphi(a)=a$ & $\varphi(a)=a$ \\
  $\psi_2(b)=aa$  & $\varphi(b)=aa$ & $\varphi(b)=aaa$ & $\varphi(b)=aaa$ & $\varphi(b)=aaaa$ & $\varphi(b)=aaaaaa$ & $\varphi(b)=aaaaa$ \\
  \hline
    $\psi_2(a)=a$ & $\varphi(a)=a$ & $\varphi(a)=a$ & $\varphi(a)=a$ & $\varphi(a)=a$ & $\varphi(a)=a$ & $\varphi(a)=a$ \\
  $\psi_2(b)=bbb$  & $\varphi(b)=bbb$ & $\varphi(b)=bbba$ & $\varphi(b)=abbb$ & $\varphi(b)=bbbbbb$ & $\varphi(b)=b^9$ & $\varphi(b)=bbbbbba$ \\
  \hline
    $\psi_2(a)=a$ & $\varphi(a)=a$ & $\varphi(a)=a$ & $\varphi(a)=a$ & $\varphi(a)=a$ & $\varphi(a)=a$ & $\varphi(a)=a$ \\
  $\psi_2(b)=bba$  & $\varphi(b)=bba$ & $\varphi(b)=bbaa$ & $\varphi(b)=abba$ & $\varphi(b)=bbabba$ & $\varphi(b)=bbabbabba$ & $\varphi(b)=bbabbaa$ \\
  \hline
    $\psi_2(a)=a$ & $\varphi(a)=a$ & $\varphi(a)=a$ & $\varphi(a)=a$ & $\varphi(a)=a$ & $\varphi(a)=a$ & $\varphi(a)=a$ \\
  $\psi_2(b)=bab$  & $\varphi(b)=bab$ & $\varphi(b)=baba$ & $\varphi(b)=abab$ & $\varphi(b)=babbab$ & $\varphi(b)=babbabbab$ & $\varphi(b)=babbaba$ \\
  \hline
    $\psi_2(a)=a$ & $\varphi(a)=a$ & $\varphi(a)=a$ & $\varphi(a)=a$ & $\varphi(a)=a$ & $\varphi(a)=a$ & $\varphi(a)=a$ \\
  $\psi_2(b)=abb$  & $\varphi(b)=abb$ & $\varphi(b)=abba$ & $\varphi(b)=aabb$ & $\varphi(b)=abbabb$ & $\varphi(b)=abbabbabb$ & $\varphi(b)=abbabba$ \\
  \hline
    $\psi_2(a)=a$ & $\varphi(a)=a$ & $\varphi(a)=a$ & $\varphi(a)=a$ & $\varphi(a)=a$ & $\varphi(a)=a$ & $\varphi(a)=a$ \\
  $\psi_2(b)=aab$  & $\varphi(b)=aab$ & $\varphi(b)=aaba$ & $\varphi(b)=aaab$ & $\varphi(b)=aabaab$ & $\varphi(b)=aabaabaab$ & $\varphi(b)=aabaaba$ \\
  \hline
    $\psi_2(a)=a$ & $\varphi(a)=a$ & $\varphi(a)=a$ & $\varphi(a)=a$ & $\varphi(a)=a$ & $\varphi(a)=a$ & $\varphi(a)=a$ \\
  $\psi_2(b)=aba$  & $\varphi(b)=aba$ & $\varphi(b)=abaa$ & $\varphi(b)=aaba$ & $\varphi(b)=abaaba$ & $\varphi(b)=abaabaaba$ & $\varphi(b)=abaabaa$ \\
  \hline
    $\psi_2(a)=a$ & $\varphi(a)=a$ & $\varphi(a)=a$ & $\varphi(a)=a$ & $\varphi(a)=a$ & $\varphi(a)=a$ & $\varphi(a)=a$ \\
  $\psi_2(b)=baa$  & $\varphi(b)=baa$ & $\varphi(b)=baaa$ & $\varphi(b)=abaa$ & $\varphi(b)=baabaa$ & $\varphi(b)=baabaabaa$ & $\varphi(b)=baabaaa$ \\
  \hline
  \hline
    \end{tabular}
  \end{adjustbox}
  \caption{Exhaustive Search Part 1 : Shows the composition $\varphi=\psi_2\circ \psi_1$. In Example \ref{ex:first irreducible}, we show that $\varphi_2$ is irreducible as we exhaustively consider all compositions of rank-preserving endomorphisms $\psi_1, \psi_2$ such that $|(\psi_2\circ\psi_1)(a)|=1$ and $|(\psi_2\circ\psi_1)(b)|\leq3$.}
  \label{tab:placeholder1}
\end{table}
\begin{table}[h!]
  \centering
  \tiny
  \begin{adjustbox}{center,max width=\linewidth}
    \begin{tabular}{||c||c|c|c|c|c||}
      \hline
      \hline
      & $\psi_1(a)=a$ & $\psi_1(a)=a$ & $\psi_1(a)=a$ & $\psi_1(a)=a$
        & $\psi_1(a)=a$   \\
      & $\psi_1(b)=bab$ & $\psi_1(b)=abb$ & $\psi_1(b)=aab$ & $\psi_1(b)=aba$
        & $\psi_1(b)=baa$  \\
      \hline
      \hline
   $\psi_2(a)=a$ & $\varphi(a)=a$ &$\varphi(a)=a$ & $\varphi(a)=a$ & $\varphi(a)=a$ & $\varphi(a)=a$ \\
  $\psi_2(b)=b$  & $\varphi(b)=bab$ & $\varphi(b)=abb$ & $\varphi(b)=aab$ & $\varphi(b)=aba$ & $\varphi(b)=baa$  \\
  \hline
  $\psi_2(a)=a$ & $\varphi(a)=a$ &$\varphi(a)=a$ & $\varphi(a)=a$ & $\varphi(a)=a$ & $\varphi(a)=a$  \\
  $\psi_2(b)=ba$  & $\varphi(b)=baaba$ & $\varphi(b)=ababa$ & $\varphi(b)=aaba$ & $\varphi(b)=abaa$ & $\varphi(b)=baaa$  \\
  \hline
    $\psi_2(a)=a$ & $\varphi(a)=a$ &$\varphi(a)=a$ & $\varphi(a)=a$ & $\varphi(a)=a$ & $\varphi(a)=a$  \\
  $\psi_2(b)=ab$  & $\varphi(b)=abaab$ & $\varphi(b)=aabab$ & $\varphi(b)=aaab$ & $\varphi(b)=aaba$ & $\varphi(b)=abaa$ \\
  \hline
    $\psi_2(a)=a$ & $\varphi(a)=a$ &$\varphi(a)=a$ & $\varphi(a)=a$ & $\varphi(a)=a$ & $\varphi(a)=a$ \\
  $\psi_2(b)=bb$  & $\varphi(b)=bbabb$ & $\varphi(b)=abbbb$ & $\varphi(b)=aabb$ & $\varphi(b)=abba$ & $\varphi(b)=bbaa$  \\
  \hline
    $\psi_2(a)=a$ & $\varphi(a)=a$ &$\varphi(a)=a$ & $\varphi(a)=a$ & $\varphi(a)=a$ & $\varphi(a)=a$  \\
  $\psi_2(b)=aa$  & $\varphi(b)=aaaaa$ & $\varphi(b)=aaaaa$ & $\varphi(b)=aaaa$ & $\varphi(b)=aaaa$ & $\varphi(b)=aaaa$  \\
  \hline
    $\psi_2(a)=a$ & $\varphi(a)=a$ &$\varphi(a)=a$ & $\varphi(a)=a$ & $\varphi(a)=a$ & $\varphi(a)=a$  \\
  $\psi_2(b)=bbb$  & $\varphi(b)=bbbabbb$ & $\varphi(b)=abbbbbb$ & $\varphi(b)=aabbb$ & $\varphi(b)=abbba$ & $\varphi(b)=bbbaa$  \\
  \hline
    $\psi_2(a)=a$ & $\varphi(a)=a$ &$\varphi(a)=a$ & $\varphi(a)=a$ & $\varphi(a)=a$ & $\varphi(a)=a$  \\
  $\psi_2(b)=bba$  & $\varphi(b)=bbaabba$ & $\varphi(b)=abbabba$ & $\varphi(b)=aabba$ & $\varphi(b)=abbaa$ & $\varphi(b)=bbaaa$ \\
  \hline
    $\psi_2(a)=a$ & $\varphi(a)=a$ &$\varphi(a)=a$ & $\varphi(a)=a$ & $\varphi(a)=a$ & $\varphi(a)=a$  \\
  $\psi_2(b)=bab$  & $\varphi(b)=bababab$ & $\varphi(b)=ababbab$ & $\varphi(b)=aabab$ & $\varphi(b)=ababa$ & $\varphi(b)=babaa$  \\
  \hline
    $\psi_2(a)=a$ & $\varphi(a)=a$ &$\varphi(a)=a$ & $\varphi(a)=a$ & $\varphi(a)=a$ & $\varphi(a)=a$  \\
  $\psi_2(b)=abb$  & $\varphi(b)=abbaabb$ & $\varphi(b)=aabbabb$ & $\varphi(b)=aaabb$ & $\varphi(b)=aabba$ & $\varphi(b)=abbaa$  \\
  \hline
    $\psi_2(a)=a$ & $\varphi(a)=a$ &$\varphi(a)=a$ & $\varphi(a)=a$ & $\varphi(a)=a$ & $\varphi(a)=a$  \\
  $\psi_2(b)=aab$  &$\varphi(b)=aabaaab$ & $\varphi(b)=aaabaab$ & $\varphi(b)=aaaab$ & $\varphi(b)=aaaba$ & $\varphi(b)=aabaa$ \\
  \hline
    $\psi_2(a)=a$ & $\varphi(a)=a$ &$\varphi(a)=a$ & $\varphi(a)=a$ & $\varphi(a)=a$ & $\varphi(a)=a$  \\
  $\psi_2(b)=aba$  & $\varphi(b)=abaaaba$ & $\varphi(b)=aabaaba$ & $\varphi(b)=aaaba$ & $\varphi(b)=aabaa$ & $\varphi(b)=abaaa$  \\
  \hline
    $\psi_2(a)=a$ & $\varphi(a)=a$ &$\varphi(a)=a$ & $\varphi(a)=a$ & $\varphi(a)=a$ & $\varphi(a)=a$  \\
  $\psi_2(b)=baa$  & $\varphi(b)=baaabaa$ & $\varphi(b)=abaabaa$ & $\varphi(b)=aabaa$ & $\varphi(b)=abaaa$ & $\varphi(b)=baaaa$ \\

  \hline
  \hline
    \end{tabular}
  \end{adjustbox}
  \caption{Exhaustive Search Part 2 : Shows the composition $\varphi=\psi_2\circ \psi_1$. In Example \ref{ex:first irreducible}, we show that $\varphi_2$ is irreducible as we exhaustively consider all compositions of rank-preserving endomorphisms $\psi_1, \psi_2$ such that $|(\psi_2\circ\psi_1)(a)|=1$ and $|(\psi_2\circ\psi_1)(b)|\leq3$.}
  \label{tab:placeholder2}
\end{table}
\end{document}